\documentclass[letterpaper,10pt,pra,aps,showpacs,twocolumn]{revtex4-1}
\usepackage{hyperref}
\usepackage{amsmath}
\usepackage{amssymb}
\usepackage{amsthm}
\usepackage{amsfonts}
\usepackage[all]{xy}
\usepackage{graphicx}
\usepackage{color}

\newtheorem{lem}{Lemma}
\newtheorem{theorem}{Theorem}
\theoremstyle{definition} \newtheorem{defn}{Definition}
\theoremstyle{definition} \newtheorem{remark}{Remark}

\newcommand{\rank}{\mathrm{rank}}
\newcommand{\supp}{\mathrm{supp}}
\newcommand{\ket}[1]{\left| #1 \right\rangle}
\newcommand{\ie}{\emph{i.e.}}

\begin{document}

\title{Local stabilizer codes in three dimensions without string logical operators}
\author{Jeongwan Haah}
\email{jwhaah@caltech.edu}
\affiliation{Institute for Quantum Information, California Institute of Technology, Pasadena, CA 91125, USA}
\date{28 February 2011}

\begin{abstract}
We suggest concrete models for self-correcting quantum memory by reporting examples of local stabilizer codes in 3D that have no string logical operators.
Previously known local stabilizer codes in 3D all have string-like logical operators, which make the codes non-self-correcting.
We introduce a notion of ``logical string segments'' to avoid difficulties in defining one dimensional objects in discrete lattices.
We prove that every string-like logical operator of our code can be deformed to a disjoint union of short segments, each of which is in the stabilizer group. The code has surface-like logical operators whose partial implementation has unsatisfied stabilizers along its boundary.
\end{abstract}

\maketitle

\section{Introduction}

Self-correcting quantum memory is an interesting subject not only because of its application for quantum information processing technology, but also because of its implication for quantum many-body physics; it shows a topological order at finite temperature. 
It is known that in 4D a self-correcting quantum memory is possible: Toric code \cite{DennisKitaevLandahlEtAl2002Topological, AlickiHorodeckiHorodeckiEtAl2008thermal}, which is a CSS stabilizer code. There are classes of models in 2D that are not self-correcting \cite{BravyiTerhal2008no-go, KayColbeck2008Quantum, HaahPreskill2010Logical} including those based on local stabilizer codes. It is thus a natural question whether a self-correcting quantum memory is possible in 3D, at least in the class of models based on stabilizer codes.

A string-like logical operator plays an important role in the thermal instability; its existence is crucial in the no-go theorems \cite{BravyiTerhal2008no-go, KayColbeck2008Quantum, HaahPreskill2010Logical} for self-correcting quantum memory in 2D based on local stabilizer codes, and more generally, on local commuting projector codes. The string-like logical operator arises easily under the interaction with thermal bath, and hence, adversely affect encoded information. Known models in 3D \emph{e.g.}, toric code~\cite{NussinovOrtiz2007Autocorrelations, CastelnovoChamon2008Topological}, Chamon model~\cite{Chamon2005Quantum}, topological color code \cite{BombinMartin-Delgado2007Exact} and Kim model \cite{Kim2010Exactly}, do have string-like logical operators. Bacon subsystem code in 3D \cite{Bacon2006Operator} which does not have string-like bare logical operator, might be self-correcting, but it is not yet affirmative since its Hamiltonian is hard to solve.

There is an issue of defining string-like logical operator. Since a lattice is a discrete space, it is generally not possible to define the dimension of a subset of the lattice. An observation is that a string is a union of segments, each of which has two end points. Thus we define logical string segments as a finite object that has two anchors at the end with its middle part commuting with stabilizer generators. A string logical operator is then a logical operator that contains arbitrarily long logical string segments.

The main result of this paper is that there exist local CSS stabilizer codes in 3D that are free of string logical operators. We give the complete classification of codes (\emph{cubic codes}) under our consideration in Sec.~\ref{section:list_cubic_codes}. We explain how we classify them in Sec.~\ref{section:comm_rel_corner_operators}. In Sec.~\ref{section:macroscopic_code_distance}, we prove that the code distance is at least linear in system size. Sec.~\ref{section:logical_string_segments} is the central section where we define a logical string segment and prove that four of our codes are free of string logical operators. We report exact empirical formulae of the number of logical qubits of our codes in Sec.~\ref{section:k_vs_L}. Finally, we discuss thermal stability of our codes and related issues in Sec.~\ref{section:thermal_stability}. Sec.~\ref{section:discussion} contains our concluding remarks.


Let us review the formalism of stabilizer codes. Let $\mathcal{P}_n$ be the group of Pauli operators acting on $n$ qubits. An abelian subgroup $\mathcal{S}$ of $\mathcal{P}_n$ is called the \emph{stabilizer group} if $ -I \notin \mathcal{S}$. The stabilizer group $\mathcal{S}$ defines a subspace of $n$-qubit Hilbert space by
\[
 \mathcal{C} = \{ \ket{\psi} : s \ket{\psi} = \ket{\psi} \text{for all } s \in \mathcal{S} \},
\]
which is the code space. $\mathcal{C}$ is nonzero because $-I \notin \mathcal{S}$. A CSS code is defined by a stabilizer group, each element of which can be written as a product of $X$- and $Z$- type stabilizer elements. The Pauli group has a nice property that any pair of elements is either commuting or anti-commuting, and that every element squares to identity. If we abelianize the Pauli group $\mathcal{P}_n$ by ignoring all phase factors~\cite{CalderbankRainsShorEtAl1997Quantum}, we obtain $2n$-dimensional vector space over the binary field equipped with a symplectic bilinear form $\lambda$; $\lambda(a,b) = 1$ if $a$ and $b$ anti-commute, and $\lambda(a,b)=0$ if they commute. The product of two Pauli operators is expressed by the addition of the two corresponding vectors. The identity operator is the zero vector.

In this respect, the stabilizer group is characterized as an isotropic subgroup. Note that the condition that $-I \notin \mathcal{S}$ should be checked separately. We abuse the notation and use the same symbol $\mathcal{S}$ to denote the vector space corresponding to the group. The orthogonal complement $\mathcal{S}^\perp$ of $\mathcal{S}$ with respect to the symplectic form is the space of logical operators. The set of nontrivial logical operators modulo stabilizer group is the quotient space $\mathcal{S}^\perp / \mathcal{S}$. A stabilizer code is (geometrically-)\emph{local} if its stabilizer group is generated by (geometrically-)local Pauli operators.

A translation-invariant local stabilizer code can be defined on the infinite lattice. In this case, we define the stabilizer group as a group of all \emph{finite} products of local generators. A logical operator is a Pauli operator possibly with infinite support that commutes with every generator. Since each generator is local, the commutation relation between the generator and an arbitrary Pauli operator is well-defined.


\section{Complete list of cubic codes}
\label{section:list_cubic_codes}

We seek for a simple local stabilizer code that is translation-invariant, encodes at least one logical qubit, has large code distance, and does not have any string logical operator. (Formal definition of string logical operator will be given in Sec.~\ref{section:logical_string_segments}.)
To start with, consider a local stabilizer code on a $D$-dimensional simple cubic lattice $\mathbb{Z}^D$ with one qubit at each site. A general stabilizer generator may act on a bounded number of qubits in an arbitrary way. However, if we coarse-grain the lattice, or equivalently, put $m \ge 1$ qubits at each site of the lattice, we can say without loss of generality that a generator acts on the qubits on $2^D$ sites of a unit hypercube. 

We focus on stabilizer codes with only two types of generators for simplicity. Since each generator can be described by a $2m \times 2^D$-component binary vector, there are a finite yet large number of conceivable generators. We will demand a certain structure of generators in order to reduce the number of candidates. We will further impose conditions such that the code does not contain any nontrivial logical operator on a ``straight line'', which will be necessary for the codes to be without string logical operator.

The structure of the generators is restricted in the following way: For CSS codes, there are two types of generators corresponding to $Z$- and $X$-type. We denote by $\alpha_i$ a corner of the cube of generator type $i$, and by $\alpha'_i$ the body-opposite corner as depicted in Fig.~\ref{fig:general_generator}. For non-CSS codes, the generators should satisfy $\alpha_1 = \alpha'_2$.

\begin{figure}
\centering
\includegraphics[width=0.35\textwidth]{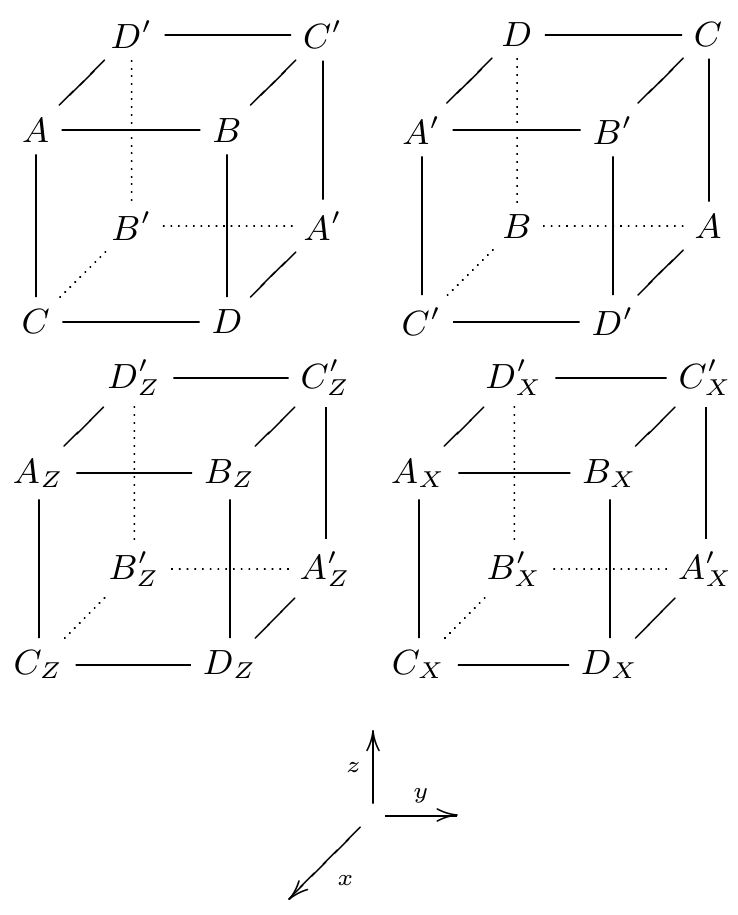}
\caption{Stabilizer Generators for non-CSS(top) and CSS(bottom) cubic codes. Throughout the paper we fix the coordinate system as shown.}
\label{fig:general_generator}
\end{figure}

The number $m$ of qubits per site should be bounded by the number of types of generators. For a local stabilizer code in any finite lattice, with open or periodic boundary conditions, there is a tradeoff in 3D \cite{BravyiPoulinTerhal2010Tradeoffs}
\[
 k d = O(L^3)
\]
between $k$, the number of logical qubits, and $d$, the code distance, where $L$ is the linear size of the lattice. If there are $t < m$ types of generators, the number of independent stabilizers is at most $t L^3$, and $k \geq (m-t)L^3$. The code distance is then a constant independent of $L$. In order to achieve macroscopic code distance, it is mandatory that $m \leq t$.

A string operator may wrap around a finite periodic lattice many times that it looks like a surface. But this is a property of boundary; as long as thermal stability is concerned, we ignore the boundary effects and consider stabilizer codes in the infinite lattice $\Lambda = \mathbb{Z}^3$.
If a single site operator $E$, \ie, two qubit operator, is logical, we want it to be an element of the stabilizer group $\mathcal{S}$. Since the stabilizer group does not explicitly include a single site operator, it is not easy to formulate the condition $E \in \mathcal{S}$. For simplicity, we require that $E$ is the identity up to phase.
For a single site operator $E$, we denote by $E[v]_p$ the Pauli operator repeated along the line parallel to $v$ passing $p$, \ie,
\begin{equation}
E[v]_p = \cdots \otimes E \otimes E \otimes \cdots,
\end{equation}
whose \emph{support}, the set of \emph{sites} on which a Pauli operator acts nontrivially, is the line,
\[
\supp (E[v]_p) = \{ p + n v \in \Lambda ~|~ n \in \mathbb{Z} \} .
\]
We say $E[v]$ has \emph{period one} if $\| v \|_\infty = 1$. ($\|(a,b,c)\|_\infty = \max\{ |a|, |b|, |c| \}$.) We demand that any logical operator of period one be the identity up to phase. This condition is not sufficient for the code to be free of string logical operators, but is necessary. We will see that a nontrivial logical operator of period one is a string logical operator in our formal definition of strings in Sec.~\ref{section:logical_string_segments}.

Imposing the constraints above may result in a trivial code in a finite lattice for which the number of encoded qubits is zero ($k=0$). To avoid such a case, we restrict the generators such that the product of all corner operators to be the identity operator up to phase. This condition is automatically satisfied by non-CSS codes under consideration. For CSS codes, this becomes a nontrivial algebraic constraint on the corner operators.

There are equivalence relations on the set of codes. If two stabilizer codes are related by a symmetry transformation of the unit cube, they are essentially the same. If one can be transformed into the other by a basis change on each site, we also regard them as the same codes. Renaming of stabilizer generators obviously gives equivalent codes. Up to these equivalences, we report that there are 1 non-CSS and 17 CSS \emph{cubic codes} listed in Table~\ref{tb:list_codes}. The conditions of the cubic codes are summarized below.
\begin{description}
 \item[Condition 1] There are one or two qubits per site in the infinite simple cubic lattice $\mathbb{Z}^3$.
 \item[Condition 2] The stabilizer group $\mathcal{S}$ is translation-invariant and is generated by two types of operators acting on eight corners of a cube. For non-CSS code, the two are related by spatial inversion, \ie, $\alpha'_1 = \alpha_2$. See Fig.~\ref{fig:general_generator}. The product of all corner operators of a CSS code is the identity.
 \item[Condition 3] If $ E \in \mathcal{S}^\perp $ is a single site operator, then $E$ is the identity up to phase.
 \item[Condition 4] If $ l \in \mathcal{S}^\perp $ has period one in $l_\infty$-metric, \ie, $\supp(l)$ is along one of 3 coordinate axes, 6 face-diagonals, or 4 body-diagonals, then $l$ is the identity up to phase.
\end{description}
In the next section, we will study the conditions systematically.

\begin{table}
\centering
\newcommand{\writegenerator}[8]{$#1$ & $#2$ & $#3$ & $#4$ & $#5$ & $#6$ & $#7$ & $#8$}
 \begin{tabular}{c|cccccccc}
\hline
\hline
 & \multicolumn{8}{c}{Corner operators $\alpha$} \\
 & \writegenerator{A}{B}{C}{D}{A'}{B'}{C'}{D'} \\
\hline
\hline
$0^\dagger$ & \writegenerator{XX}{ZI}{ZY}{XY} {ZZ}{II}{XZ}{ZX} \\
\hline
$1^\star$   & \writegenerator{ZI}{ZZ}{IZ}{ZI} {IZ}{II}{ZI}{IZ} \\ 
$2^\star$   & \writegenerator{IZ}{ZZ}{ZI}{ZI} {ZI}{ZZ}{IZ}{ZI} \\ 
$3^\star$   & \writegenerator{IZ}{ZZ}{ZZ}{ZI} {ZZ}{II}{IZ}{IZ} \\ 
$4^\star$   & \writegenerator{IZ}{ZZ}{ZI}{ZI} {IZ}{II}{IZ}{ZI} \\ 
$5$         & \writegenerator{ZI}{ZZ}{II}{ZZ} {ZI}{II}{IZ}{IZ} \\ 
$6$         & \writegenerator{ZI}{II}{ZI}{ZZ} {IZ}{ZZ}{II}{IZ} \\ 
$7$         & \writegenerator{ZI}{ZZ}{ZI}{IZ} {IZ}{II}{II}{ZZ} \\ 
$8$         & \writegenerator{ZI}{ZI}{IZ}{ZZ} {IZ}{II}{IZ}{ZI} \\ 
$9$         & \writegenerator{ZI}{IZ}{ZZ}{ZZ} {IZ}{ZZ}{II}{IZ} \\ 
$10$        & \writegenerator{ZI}{IZ}{ZI}{ZZ} {IZ}{ZZ}{ZI}{ZI} \\ 
$11^\dagger$& \writegenerator{ZI}{ZZ}{II}{IZ} {ZI}{II}{IZ}{ZZ} \\ 
$12^\dagger$& \writegenerator{ZI}{IZ}{ZZ}{ZZ} {ZI}{II}{II}{IZ} \\ 
$13^\dagger$& \writegenerator{ZI}{ZZ}{IZ}{ZI} {IZ}{II}{II}{ZZ} \\ 
$14^\dagger$& \writegenerator{ZI}{IZ}{ZZ}{ZZ} {IZ}{II}{ZZ}{IZ} \\ 
$15^\dagger$& \writegenerator{ZI}{IZ}{II}{ZZ} {IZ}{ZZ}{II}{ZI} \\ 
$16^\dagger$& \writegenerator{ZI}{ZI}{II}{IZ} {IZ}{ZZ}{II}{ZZ} \\ 
$17^\dagger$& \writegenerator{ZI}{ZZ}{IZ}{ZI} {IZ}{ZI}{ZI}{ZZ} \\ 
\hline
\hline

\end{tabular}
\caption{Complete list of cubic codes. The corners of the unit cube are labeled as in Fig.\ref{fig:general_generator}. The second generator of non-CSS Code 0 is given by the spatial inversion about body-center. The rest are all CSS codes, for which $X$-type generator is uniquely determined by eq.~\eqref{eq:generator_inversion_symmetry}. The codes marked with $\star$ do \emph{not} have string logical operators, while those with $\dagger$ do. See Theorem~\ref{thm:Code1234are_free_of_strings} and Appendix~\ref{section:codes_w_string_logical_operators} }
\label{tb:list_codes}
\end{table}

\section{Commutation relations of corner operators}
\label{section:comm_rel_corner_operators}

Given a set $\{ g_1, \ldots, g_n \} \subseteq \mathcal{P}_m$ of $n$ Pauli operators acting on $m$ qubits, we can express their commutation relations in an $n \times n$ skew-symmetric (and, at the same time, symmetric) matrix $\omega$ over the binary field.
\[
\omega_{ij} = g_i^T \lambda g_j
\]
If we express $g_i$'s in the columns of a $2m \times n$ matrix $P$, then obviously
\[
  P^T \lambda P = \omega .
\]
Since $\omega$ is skew-symmetric, $r \equiv \rank(\omega)$ is even. Note that the rank of $\lambda$ is $2m$. Since $ r \leq \min \{ \rank(P), \rank(\lambda) \} $, we see that $r/2$ is the minimum possible number of qubits on which Pauli operators of $P$ act. Conversely,
\begin{lem}
 Given a commutation relation $\omega$ of Pauli operators, all realizations $P$ of $\omega$ using minimum number of qubits are equivalent up to symplectic transformations.
\label{lem:unique_realization_of_commutation_relation}
\end{lem}
\begin{proof}
The rank of $P$ is at least $r = 2m$. Being of full rank, $P$ has linearly independent rows and we can add extra $n-2m$ rows to $P$ so that the extension $P^e$ of $P$ is invertible. Let $P_1, P_2$ be two solutions realizing $\omega$, and $P^e_1 , P^e_2$ be their arbitrary invertible extensions respectively. We have
\[
 (P_1^e)^T \lambda^e P_1^e = \omega = (P_2^e)^T \lambda^e P_2^e 
\]
where
\[
\lambda^e = \begin{pmatrix}\lambda & 0 \\ 0 & 0  \end{pmatrix}.
\]
Therefore, $S^e = P_1^e (P_2^e)^{-1}$ is a symplectic transformation preserving $\lambda^e$. The most general form of a transformation preserving $\lambda^e$ is
\[
S^e = \begin{pmatrix} S & 0 \\ C & D \end{pmatrix},
\]
where $S$ is such that $\lambda = S^T \lambda S$. Immediately, $P_1^e = S^e P_2^e$, or $ P_1 = S P_2 $.
\end{proof}

We will translate all the requirements for the cubic codes into conditions on commutation relation matrix $\omega$ of corner operators.
First, $\omega$ must represent a stabilizer code. The generators at different locations will commute if the components of $\omega$ satisfies a certain linear equation. We must consider all the cases when two generators meet with each other at a site, at an edge, at a face, and when they overlap completely. Note that this classification is based on Condition 2.
For non-CSS codes, the equations are
\begin{align*}
\omega(A,A')& = 0 , & \omega(B,B')& = 0 , \\
\omega(C,C')& = 0 , & \omega(D,D')& = 0 
\end{align*}
for the generator meeting at a site,
\begin{align*}
\omega(A,C') + \omega(C,A')& = 0 , & \omega(B,D') + \omega(D,B')& = 0 , \\
\omega(A,B') + \omega(B,A')& = 0 , & \omega(C,D') + \omega(D,C')& = 0 , \\
\omega(C,B) + \omega(B',C')& = 0 , & \omega(A,D) + \omega(D',A')& = 0 
\end{align*}
for them meeting at an edge,
\begin{align*}
\omega(A,B) + \omega(C,D) + \omega(B',A') + \omega(D',C')& = 0 , \\
\omega(A,C) + \omega(B,D) + \omega(D',B') + \omega(C',A')& = 0 , \\
\omega(A,D') + \omega(B,C') + \omega(C,B') + \omega(D,A')& = 0 
\end{align*}
for them meeting at a face. When generators meet at a cube, they automatically commute. Above 13 equations are independent since each term, \emph{e.g.} $\omega(A,A')$, appears only in one equation. For CSS codes, we only need to consider commutation relation between $X$-type and $Z$-type. There are 8 equations for the case when two generators meet with each other at a site, 12 equations when they meet at an edge, 6 equations when they meet at a face, and 1 equation when they overlap completely. One easily checks that these 27 equations are independent since each term appears only once.

Consider a non-CSS cubic code for which $\omega$ is $8 \times 8$. Let $m \geq 1$ be the number of qubits per site, and $P$ be the $2m \times 8$ matrix whose columns consist of corner operators. A single site operator $E$ is logical if and only if $P^T \lambda E = 0$. Given $E$ logical, we require $E = 0$, or $P^T \lambda$ have the trivial kernel (Condition 3). In other words, $\rank(P)=2m$. Therefore, we must have
\begin{equation}
 r = \rank ~ (\omega) = 2m .
\label{eq:no_single_site_logical_condition}
\end{equation}
Note that it implies $m$ be the minimum possible realizing $\omega$. Conversely, if $r=2m$, we need at least $m$ qubits to realize $\omega$, and $\rank(P) \geq 2m$. Therefore, $P^T \lambda$ has trivial kernel, and there is no nontrivial logical operator supported on a single site.

Consider a logical operator $E[\hat{y}]$ of period one along $y$-axis. $E \otimes E$ commute with $A \otimes B$ if and only if 
\[
 \lambda(A \otimes B, E \otimes E) = \lambda(A,E) + \lambda(B,E) = \lambda(AB,E) = 0.
\]
See Fig.~\ref{fig:general_generator}. Hence, $E[\hat{y}]$ is logical if and only if
\begin{align*}
\lambda(AB, E)   &=0   & \lambda(CD, E)   &=0 \\
\lambda(B'A', E) &=0   & \lambda(D'C', E) &=0
\end{align*}
These equations form a system of linear equations with unknown $E$, a $2m$-component column vector. The coefficient matrix $M$ is a $4 \times 2m$ matrix.
\[
M = R P^T \lambda
\]
where $A$ is expressed in the first row of $P^T$, $B$ in the second, etc, and
\begin{equation}
R = 
\begin{pmatrix}
1&1 & 0&0 & 0&0 & 0&0 \\
0&0 & 1&1 & 0&0 & 0&0 \\
0&0 & 0&0 & 1&1 & 0&0 \\
0&0 & 0&0 & 0&0 & 1&1 
\end{pmatrix}.
\label{eq:deriving_matrix}
\end{equation}
Given $M E = 0$, we require $E=0$ (Condition 4). By eq.\eqref{eq:no_single_site_logical_condition}, we see that $P$ can be extended to be invertible, so that $\rank( R P^T \lambda ) = \rank (R (P^e)^T \lambda^e P^e ) = \rank ( R \omega )$. The requirement becomes a simple formula:
\begin{equation}
\rank \ ( R \omega ) = 2m .
\end{equation}
Conversely, if $\rank(R \omega)=2m$, then there is no logical operator of period one along $y$-axis. In an analogous manner, we consider all 13 logical operators along lines that are respectively parallel to 3 coordinate axes, to 6 face-diagonals, and to 4 body diagonals. Note that the `derived' matrix $R \omega$ is calculated by adding rows of $\omega$ corresponding to corners that the logical operator of period one passes through.

A CSS cubic code has 16 corner operators. The corners belonging to one of generators automatically commute with each other. Therefore, $\omega$ has non-zero elements in off-diagonal blocks if we order the corner operators as $\{ A_Z, B_Z ,\ldots, A_X, B_X, \ldots \}$:
\[
 \omega =
\begin{pmatrix}
 0 & \omega' \\
 \omega'^T & 0
\end{pmatrix}
\]
Let us also order the basis of Pauli group $\mathcal{P}_m$ such that $Z$-operators come first and
\[
 \lambda =
\begin{pmatrix}
 0 & I \\
 I & 0
\end{pmatrix} .
\]
The triviality of a single site operator is expressed as $ \rank( \omega ) = 2m$, or $\rank( \omega' ) = m$. Consider $X$-type logical operator $x[\hat{y}]$. It is logical if and only if
\begin{align*}
\lambda(A_Z  B_Z,  x) &=0 &
\lambda(C_Z  D_Z,  x) &=0 \\
\lambda(B'_Z A'_Z, x) &=0 &
\lambda(D'_Z C'_Z, x) &=0
\end{align*}
which is equivalent to a matrix equation
\begin{equation}
 R P_Z^T x = 0
\label{eq:no_lineop_matrix}
\end{equation}
where $x$ is an $m$-component vector, $A_Z$ is expressed in the first row of $P_Z^T$, etc, and $R$ is given by eq.\eqref{eq:deriving_matrix}. Since $\omega' = (P_Z)^T P_X$ has rank $m$, $P_Z$ and $P_X$ are both of full rank $m$, and we can extend them to be $P_Z^e, P_X^e$ that are invertible. Since
\[
\rank(R P_Z^T)=\rank \left(  R (P^e_Z)^T \begin{pmatrix} I & 0 \\ 0 & 0 \end{pmatrix} P^e_X \right),
\]
the matrix equation \eqref{eq:no_lineop_matrix} is equivalent to 
\begin{equation}
 m = \rank ( R \omega') .
\end{equation}
As in the non-CSS case, there are 12 more equations ensuring the triviality of the $X$-type logical operator of period one. For $Z$-type logical operators, the equations are of form
\begin{equation}
 m = \rank ( R \omega'^T) .
\end{equation}

We point out that it is a property of $\omega$ whether or not the product of all corner operators yield the identity by Lemma~\ref{lem:unique_realization_of_commutation_relation}. We have shown that the triviality of the single site operator (Condition 3) implies that any cubic code is a minimal realization of its commutation relation matrix $\omega$. Therefore, any two cubic codes with the same $\omega$ are related by a symplectic transformation, which is, in particular, an invertible linear map. The product of all the corner operators of one code is zero (\ie, the identity), if and only if the product of all the corner operators of the other is zero.

We thus completed the translation of the conditions for cubic codes into those for the commutation relation matrix of the corner operators. An advantage of this approach is that we are classifying the codes up to symplectic transformation on sites. Moreover, the cases $m=1$ and $m=2$ are treated simultaneously.

There are $2^{n_\mathrm{CSS}}$ $\omega$'s of CSS codes that are consistent with the condition that the generators define a stabilizer codes, and $2^{n_\mathrm{non-CSS}}$ $\omega$'s of non-CSS codes, where $ n_\mathrm{CSS} = 8\cdot8 - 27 = 37 $ and $ n_\mathrm{non-CSS} = {_8 C _2} - 13 = 15 $. After exhaustive search, we found no instance of $\omega$ satisfying the conditions when $m=1$. Up to the symmetry group of the unit cube and renaming of generators for CSS codes ($\omega' \leftrightarrow \omega'^T$), we finally obtain Table~\ref{tb:list_codes}. Hereafter, we will call each code by Code 0, Code 1, etc., according to Table~\ref{tb:list_codes}.

The generators of CSS cubic codes show additional symmetries that we did not impose. Namely, $Z$-type generators and $X$-type generators are related by spatial inversion. Recall that $\alpha_X, \alpha_Z, \alpha'_X, \alpha'_Z$ ($\alpha = A,B,C,D$) denote the corner operators, each of which is a 2-qubit operator. Since they are purely $Z$- or $X$-type, we express them by 2-component binary column vectors. For example, $A_Z = ZI = (1 \ 0 )^T$ of Code 1. We observe the following rule:
\begin{align}
 \alpha'_X &= \begin{pmatrix} 0 & 1 \\ 1 & 0 \end{pmatrix} \alpha_Z , &
 \alpha_X  &= \begin{pmatrix} 0 & 1 \\ 1 & 0 \end{pmatrix} \alpha'_Z
\label{eq:generator_inversion_symmetry}
\end{align}
for all $\alpha = A,B,C,D$. Because of this rule, there is a duality between $X$- and $Z$-type logical operators. That is, given the fixed origin of the lattice, for every $X$-type logical operator $O$, there exists a unique $Z$-type logical operator obtained by the spatial inversion about the origin followed by the symplectic transformation on each site defined by eq.\eqref{eq:generator_inversion_symmetry}. Hence we will consider only $X$-type logical operators and $Z$-type stabilizer generators for CSS cubic codes.

\section{Macroscopic code distance}
\label{section:macroscopic_code_distance}

In this section, we prove
\begin{theorem}
Let $d$ be the code distance of Code 0,1,2,3, or 4 defined on the periodic finite lattice $\mathbb{Z}_L^3$. Then $ d \geq L $.
\label{thm:macroscopic_code_distance}
\end{theorem}
\noindent
We introduce an important technique to deform a logical operator of cubic codes, which will prove the theorem. The technique depends on eq.\eqref{eq:generator_inversion_symmetry}. We say a Pauli operator is \emph{finite} if its support is bounded (\ie, finite set). We will prove that any finite logical operator is a product of finitely many stabilizer generators. This implies Theorem~\ref{thm:macroscopic_code_distance} by the following lemma:
 
\begin{lem}
Let $\mathcal{C}(L)$ be a translation-invariant local stabilizer code of interaction range $r > 1$ (\ie, each generator is contained in an $r^D$ hypercube) defined on a lattice $(\mathbb{Z}_L)^D$ with periodic boundary conditions, where $D$ is the dimension. Let $d=d(L)$ be the code distance of $\mathcal{C}(L)$. If there exists $L_0$ such that $d(L_0) < L_0 / (r-1)$, then there exists a finite logical operator in the infinite lattice that is not a product of finitely many stabilizer generators.
\label{lem:trivial_finite_lop_vs_d}
\end{lem}

We need a notion of connectedness:
\begin{defn}
A set of sites $ \{ p_1, p_2,\ldots, p_n \} $ is a \emph{path} joining $p_1$ and $p_n$ if for each pair $(p_i, p_{i+1})$ of consecutive sites there exists a stabilizer generator that acts nontrivially on the pair simultaneously, for $i=1,\ldots, n-1$. A set $M$ of sites is \emph{connected} if every pair of sites in $M$ are joined by a path in $M$. A \emph{connected Pauli operator} is a Pauli operator with connected support.
\end{defn}

\noindent
For example, $\{ (0,0,0), (1,0,0) \}$ is connected with respect to Code 0,1,2,3,4. $\{(0,0,0),(1,1,1)\}$ is connected with respect to Code 2, but not connected with respect to Code 0,1,3,4. See Fig.~\ref{fig:generator_Code_01234}. If a Pauli operator $O$ is logical then any of its connected component is logical. If $O$ is nontrivial, at least one of its connected components is nontrivial.

\begin{proof}(of Lemma~\ref{lem:trivial_finite_lop_vs_d})
Suppose $d = d(L_0) < L_0 / (r-1)$ for some $L_0$. There exists a connected nontrivial logical operator $O$ of support $M$, where the number of sites in $M$ is $d$. Given a closed path $\{ p_1, p_2,\ldots, p_n, p_{n+1}=p_1\} \subseteq M $, the union of the shortest line segments $c_i : [0,1] \to T^3$ connecting $c_i(0)=p_i$ and $c_i(1)=p_{i+1}$ is a trivial homological cycle of the $D$-torus $T^D \supset (\mathbb{Z}_{L_0})^D $. If it is not the case, since any nontrivial homological cycle of $T^D$ has length $L_0$, we must have $d (r-1) \geq L_0$.

Consider a lifting of all closed paths of $M$ into the universal covering $\mathbb{R}^D$ of $T^D$ via line segments $c_i$'s. Since any closed path in $M$ can always be express by a trivial homological cycle, the lifting is bounded. The corresponding lifting of $O$ is not a product of finitely many stabilizer generators, since it was not trivial.
\end{proof}
\noindent
The converse of Lemma~\ref{lem:trivial_finite_lop_vs_d} could be an interesting problem, since it, if true, implies that the code distance of translation-invariant local stabilizer code in periodic finite lattice is either $O(1)$ or $\Omega(L)$.

Consider a finite logical operator $O$. We will show that $O$ is a finite product of stabilizer generators, \ie, a trivial logical operator. We may assume $O$ is supported on a finite box $B \subseteq \Lambda$, where
\[
 B = \{ (x,y,z) ~|~ x_0 \le x \le x_1, ~y_0 \le y \le y_1,~ z_0 \le z \le z_1 \}.
\]

\begin{figure}
\centering
\includegraphics[width=0.37\textwidth]{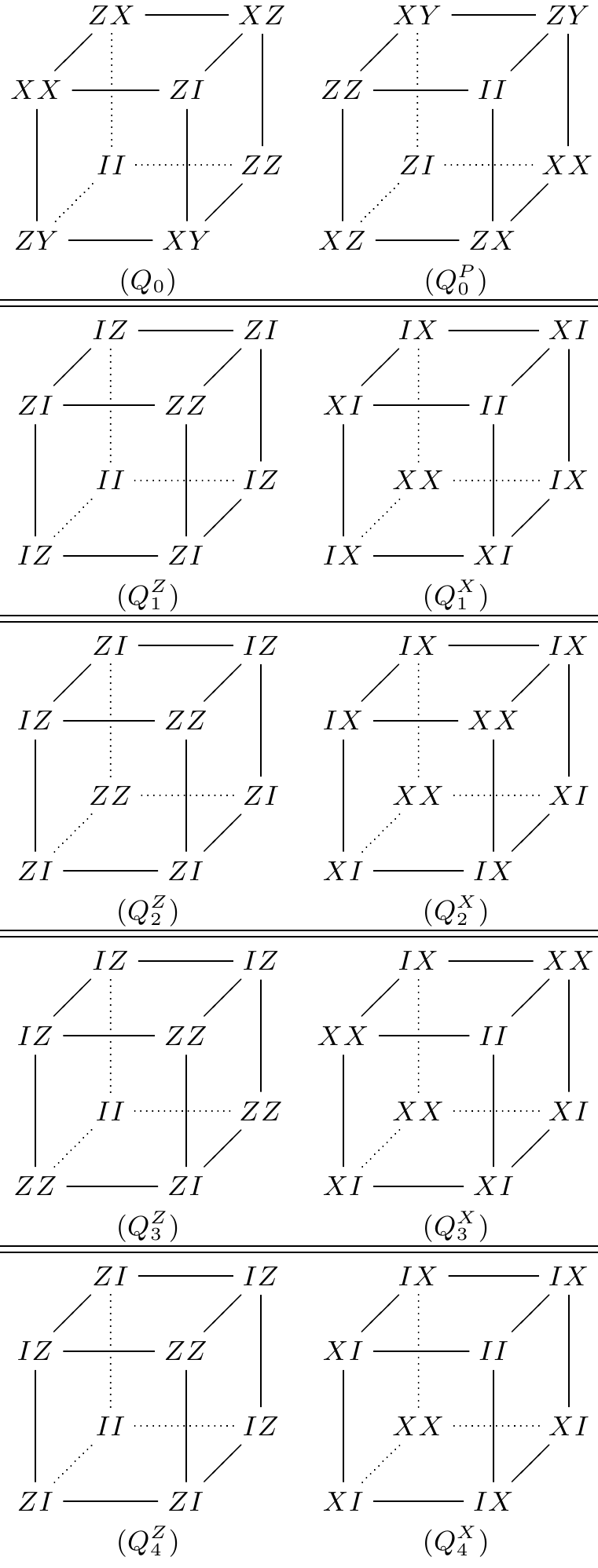}
\caption{Stabilizer generators for non-CSS Code 0, and CSS Code 1, 2, 3, and 4. They all have code distance $\ge L$ (Theorem~\ref{thm:macroscopic_code_distance}). The bottom four are free of string logical operators. See Section~\ref{section:logical_string_segments}.}
\label{fig:generator_Code_01234}
\end{figure}

We first deal with Code 0. The stabilizer generators are depicted in Fig.~\ref{fig:generator_Code_01234}.
Consider the vertex $v=(x_1,y_1,z_0)$ of $B$ that has largest $x$- and $y$-coordinate and the smallest $z$-coordinate. It must commute with $ZX$ of $Q_0$ and $XY$ of $Q_0^P$. Since $ZX = (1001)$ and $XY=(0111)$ are independent, the commutation gives two constraints on $v$. A possible $v$ is a linear combination of $ZX$ and $XY$. (Recall that Pauli group is abelianized to be a vector space.) If $v = II$, then we can shrink $B$, the support of $O$. If $v=ZX$, we can multiply $Q_0^P$ inside $B$ to make $v=II$, hence shrink $B$. If $v=XY$, then $Q_0$ inside $B$ will make $v=II$. If $v=YZ$, then the product $Q_0 Q_0^P$ inside $B$ will make $v=II$. In short, we have deformed the support $B$ of $O$ such that $B$ now consists of one less site. See the second figure of Fig.~\ref{fig:proof_macroscopic_distance}.

The process can be done arbitrarily many times as long as the deformed $B$ can contain a unit cube so that the multiplication by $Q_0$, $Q_0^P$ or both only affects the sites in $B$. Since we started with the finite box, we end up with a support consisted of three thin rectangles (the third of Fig.~\ref{fig:proof_macroscopic_distance}). To be precise, a thin rectangle $R_i$ perpendicular to $i$-axis means the set of sites
\begin{align*}
R_x &= \{ (x_0,y,z) ~|~ y_0 \le y \le y_1,~ z_0 \le z \le z_1 \} , \\
R_y &= \{ (x,y_0,z) ~|~ x_0 \le x \le x_1,~ z_0 \le z \le z_1 \} , \\
R_z &= \{ (x,y,z_1) ~|~ x_0 \le x \le x_1,~ y_0 \le y \le y_1 \} .
\end{align*}

Consider the vertex $v'= (x_1,y_0,z_0)$ of $R_y$ that is not contained in the other two thin rectangles. It must commute with $ZX, XZ$ of $Q$ and $XY, ZY$ of $Q^P$. Therefore, $v' = II$ (the fourth of Fig.~\ref{fig:proof_macroscopic_distance}). Continuing, we deduce that whole rectangle $R_y \setminus (R_x \cup R_z)$ is the identity. Note that this procedure was possible because we were able to find an edge that has ``sufficiently independent'' corner operators. We call an edge is \emph{good for erasing} if the argument above works.
Similarly, one can show that the other two rectangles $R_x$ and $R_z$ are also the identity. (The edge corresponding to $XX-ZX$ of $Q_0$ and $ZZ-XY$ of $Q_0^P$ is good for erasing, etc.) Thus, we have shown that by multiplying appropriate stabilizer generators inside $B$, we get the identity operator.

Secondly, let us show that Code 1 has macroscopic code distance.
It suffices to consider $X$-type logical operators. Let $B$ be a finite box that supports an $X$-type logical operator. Consider the vertex $v=(x_1,y_1,z_0)$ of $B$ that has largest $x$- and $y$-coordinate and the smallest $z$-coordinate. $v$ commutes with $IZ$ of $Q_1^Z$, and hence is either $II$ or $XI$. If $v=II$ we shrink $B$ by one site. If $v=XI$ we multiply $Q_1^X$ inside $B$ to erase $v$. We again end up with three thin rectangles. Consider the vertex $v'=(x_1,y_0,z_0)$ on the rectangle $R_y$ perpendicular to $y$-axis that is not contained in the other two rectangles. It commutes with $IZ, ZI$ of $Q_1^Z$. Therefore $v'=II$. Continuing, we erase $R_y$. Similarly, one can erase the other two rectangles.
Note that for CSS cubic codes, an edge is good for erasing if the corner operators of $Q^Z$ at the ends of the edge are independent.

This strategy is good enough to show that the code distance is macroscopic for Code 2, 3, and 4. We summarize the erasing procedure in Fig.~\ref{fig:proof_macroscopic_distance}.

\begin{figure}
\centering
\includegraphics[width=0.45\textwidth]{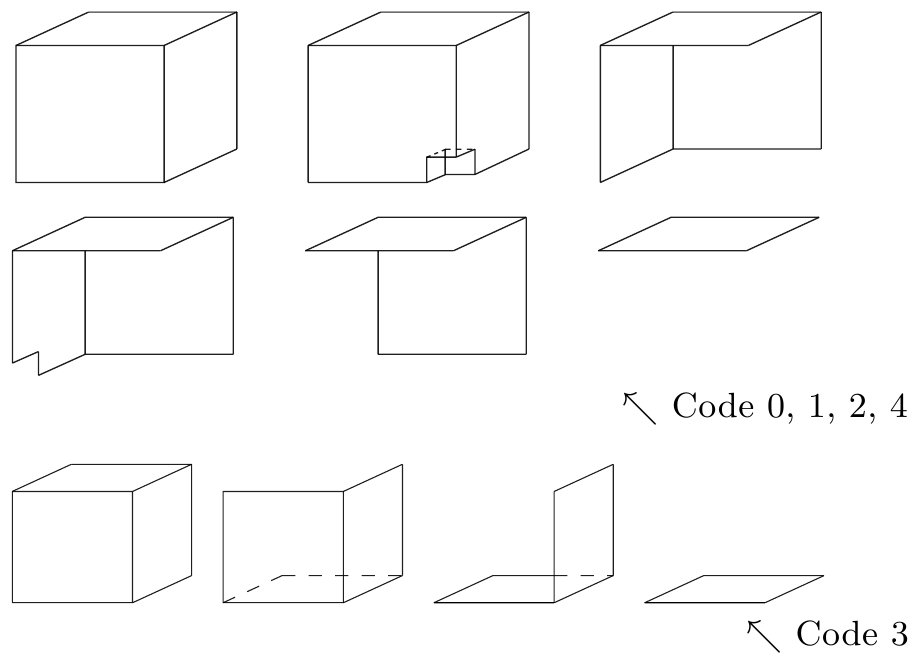}
\caption{Proof of macroscopic code distance. Deformation of a finite logical operator is depicted for each Code. For Code 1,2,3,4, the logical operator of $X$-type is considered.}
\label{fig:proof_macroscopic_distance}
\end{figure}

\section{Logical string segments}
\label{section:logical_string_segments}

A string logical operator might be regarded as a logical operator whose support is one dimensional. Indeed, the logical operators of some codes have definite topological structure. For the 2D Ising model or the toric code \cite{Kitaev2003Fault-tolerant, CastelnovoChamon2008Topological, DennisKitaevLandahlEtAl2002Topological}, the syndrome corresponding to a single site error has particular shape, by which we endow the distribution of Pauli operators with topological meaning. Concretely, $X$-error on 2D Ising model can be represented by a square in the dual lattice, and $Z$-error on 2D toric code by a link in the real lattice.

However, the topological meaning of an operator may not always be well-defined. The most important property of string logical operators would be that it can be extended in the infinite lattice to a arbitrarily long string with constant width. It is in fact the property that is used in no-go theorems on quantum memory based on stabilizer codes in 2D~\cite{BravyiTerhal2008no-go,KayColbeck2008Quantum,HaahPreskill2010Logical}. We capture this property of string logical operator in the definition as follows.

\begin{defn}
Let $\Omega_1, \Omega_2 \subset \Lambda$ be congruent cubes consisting of $w^3$ sites, and $O$ be a finite Pauli operator. A triple $\zeta=(O,\Omega_1,\Omega_2)$ is a \emph{logical string segment} if every stabilizer generator that acts trivially (by identity) on both $ \Omega_1$ and $\Omega_2 $, commutes with $O$. We call $\Omega_{1,2}$ the \emph{anchors} of $\zeta$. The \emph{directional vector} of $\zeta$ is the relative position of $\Omega_1$ to $\Omega_2$. The \emph{length} of $\zeta$ is the $l_1$-length of the directional vector, and the \emph{width} is $w$.
\end{defn}

Since $\Omega_1$ and $\Omega_2$ are congruent, the directional vector of a logical string segment is well-defined. If the directional vector is $(a,b,c)$, its $l_1$-length is $|a|+|b|+|c|$. As an example, a string operator of 2D toric code that creates a pair of vortex excitations is a logical string segment with anchors being the plaquettes carrying the vortex. Note that $O$ may not commute with all stabilizers that acts trivially on the anchors. For $O$ with $\supp(O)$ contained in the anchors is obviously a logical string segment. We need to exclude such a trivial case.

\begin{defn}
A logical string segment $\zeta=(O,\Omega_1,\Omega_2)$ is \emph{connected} if there exist two sites $p_1 \in \Omega_1, p_2 \in \Omega_2$ that can be joined by a path in $\supp(O) \cup \{p_1,p_2\}$, where $\supp(O)$ is the set of sites on which $O$ acts nontrivially. Two logical string segments $(O,\Omega_1,\Omega_2), (O',\Omega_1,\Omega_2)$ are \emph{equivalent} if $O'$ is obtained from $O$ by multiplying finitely many stabilizer generators. $\zeta$ is \emph{nontrivial} if every equivalent logical string segment is connected.
\end{defn}

For finite $w$, define $ \phi(w) $ to be the \emph{maximum length of all nontrivial logical string segments of width $w$}. $\phi$ is a non-decreasing function on the set of positive integers. Obvious from the definition is that $\phi(w) \geq 3(w-1)$ for any stabilizer code in 3D since any logical string segment with overlapping anchors is always connected, and hence nontrivial. We allow $\phi(w)$ to assume infinite value. For example, $\phi$ becomes infinite at small values of $w$ for 2D toric code~\cite{Kitaev2003Fault-tolerant}, 3D toric code~\cite{CastelnovoChamon2008Topological}, and Chamon model~\cite{Chamon2005Quantum,BravyiLeemhuisTerhal2010Topological}. String logical operator is defined as a logical operator in the infinite lattice containing an arbitrarily long nontrivial logical string segment.

\begin{defn}
A translation-invariant stabilizer code defined by a set of local stabilizer generators, 
is \emph{free of string logical operators} if the maximum length $\phi(w)$ of nontrivial logical string segment is finite for all finite $w$.
\end{defn}

2D Ising model is free of $X$-type string logical operators according to our definition. Consider an $X$-type logical string segment $\zeta=(O,\Omega_1,\Omega_2)$. Being finite, $O$ cannot be supported outside the anchors. Therefore, $\phi(w) = 2(w-1) < \infty$.
 
\begin{theorem}
Code 1, 2, 3, and 4 are free of string logical operators.
\label{thm:Code1234are_free_of_strings}
\end{theorem}

\noindent
We will argue that if the length of a logical string segment is sufficiently larger than its width, 
then there exists an equivalent logical string segment that is disconnected.
The rest of this section is the proof of Theorem~\ref{thm:Code1234are_free_of_strings}.

\subsection{Reduction to flat segments}

Let $\zeta$ be a logical string segment of width $w$. Using the technique introduced in Section~\ref{section:macroscopic_code_distance}, we will deform $\zeta$ such that it is a union of at most three logical string segments whose directional vectors are $(a,0,0),(0,b,0),(0,0,c)$ respectively; $(a,b,c)$ is the directional vector of $\zeta$. We assume $\zeta$ is of $X$-type.

Consider $\zeta$ of Code 1 such that $ a,b,c > 0 $.
Since the support of $\zeta$ is finite, we can shrink it by multiplying stabilizer generators until $\zeta$ is contained in the smallest box $B$ containing the anchors of $\zeta$. Note that this is possible due to a special property of Code 1. Namely, using the technique of Section~\ref{section:macroscopic_code_distance}, one can shrink the size ($>1$) of a finite logical operator by $1$ from any direction, since there are two orthogonal edges that are good for erasing on each of six faces.

\begin{figure}
\centering
\includegraphics[width=0.45\textwidth]{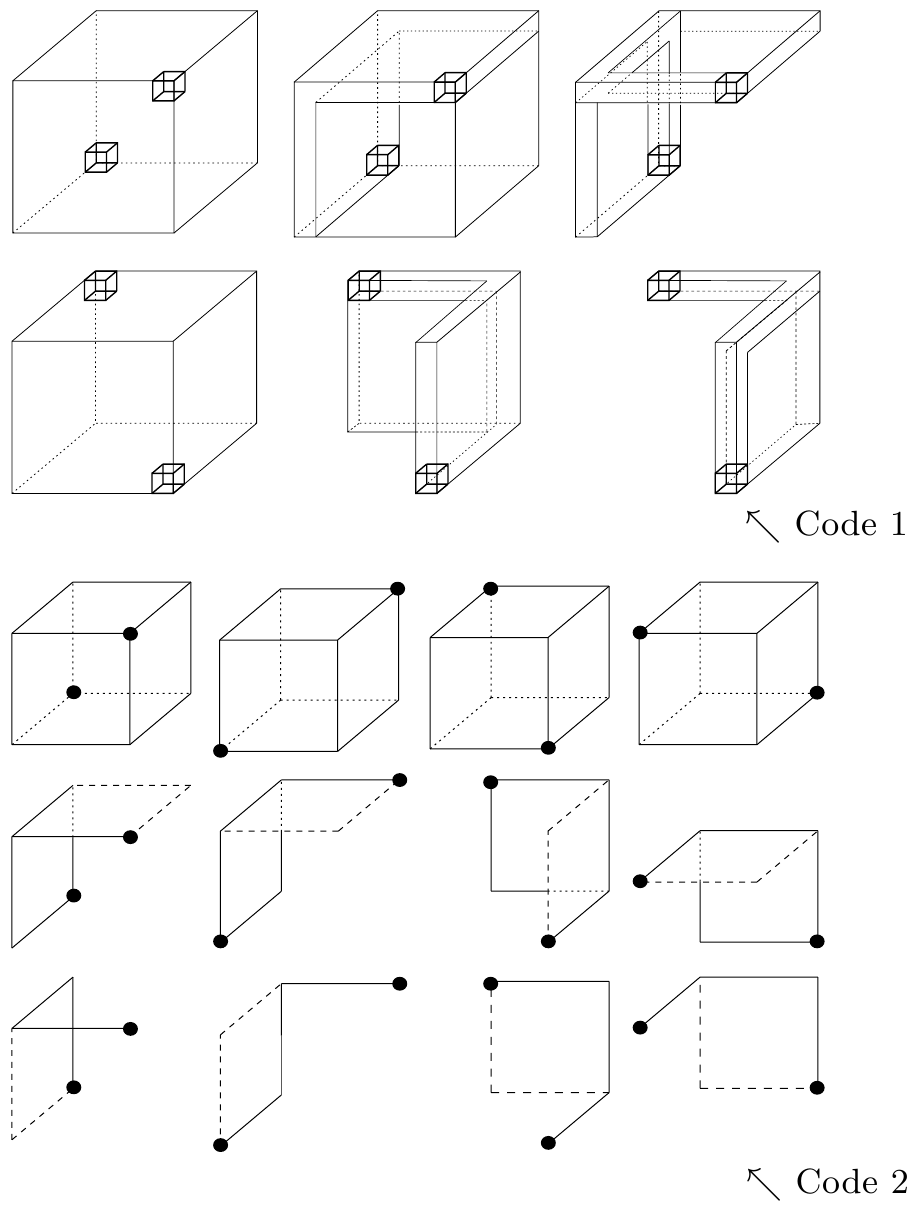}
\caption{Deformations of $X$-type logical string segments of Code 1 and 2. The small cubes and the filled dots are the anchors. The good edges for erasing are easily identified from the figures.}
\label{fig:reduction_flat_ss_code_12}
\end{figure}

Now the two anchors are located on $(0,0,0)$ (after shift of the origin) and $(a,b,c)$, and the $B$ is of size $ (a+w) \times (b+w) \times (c+w) $. Using the fact that $ZI-IZ$ along $x$-axis of $Q_1^Z$ (Fig.~\ref{fig:generator_Code_01234}) has two independent corner operator, we shrink $B$ as shown in the second figure in the first row of Fig.~\ref{fig:reduction_flat_ss_code_12}. Similarly, one shrinks $B$ further using $IZ-ZI$ (read rightward) along $y$-axis and $ZI-IZ$ (read downward) along $z$-axis. The initial $\zeta$ has been deformed such that it is a union of three logical string segments that are parallel to the coordinate axes. One can easily extend the argument to the case where $ a,b,c \geq 0 $.

Observe that Code 1 has three-fold rotational symmetry about $(1,1,1)$-axis. Therefore, we only need to consider one more case where $ a ,b \geq 0 $ and $ c \leq 0 $. The strategy is the same as before. We finish the reduction of logical string segments of Code 1 by drawing the second row of Fig.~\ref{fig:reduction_flat_ss_code_12}.

Code 2, 3, and 4 are treated similarly. For Code 2, the initial reduction of an arbitrary logical string segment to the smallest box that contains the two anchors is possible because there is a pair of orthogonal edges that are good for erasing on each of six faces of $Q_2^Z$. The subsequent reduction to flat segments is depicted in Fig.~\ref{fig:reduction_flat_ss_code_12}. Code 3 and 4 needs more explanation. See Appendix~\ref{section:pf_thm_code_34}.

Next, we show that each flat logical string segment is equivalent to a disconnected one if it is long relative to its width.

\subsection{Confusing constraints}

Consider an $X$-type logical string segment $\zeta^{(1)}_y=(O,\Omega_1,\Omega_2)$ of width $w$ of Code 1 whose directional vector is $(0,l,0)$, where $l$ is the length of $\zeta^{(1)}_y$, and $O$ is supported on the smallest box $B$ that contains $\Omega_1,\Omega_2$. Place $Q_1^Z$ such that it touches exactly two consecutive sites $p_1 = (x,y,z),~ p_2 = (x,y+1,z) \in B \setminus (\Omega_1 \cup \Omega_2)$ where $x$ is the largest and $z$ the smallest. $Q_1^Z$ acts on $p_1 - p_2$ by $IZ-ZI$, which gives a constraint on possible Pauli operators on $p_1 - p_2$ since $O$ is commuting with $Q_1^Z$. Explicitly, $p_1 - p_2$ is a linear combination (with coefficients in $\mathbb{F}_2$) of $II-IX$, $XI-II$, and $IX-XI$. We observe that given any possible operators on $p_1 - p_2$ one can make them $II-II$ by multiplying $Q_1^X$'s inside $B$. It is an important and common property of a good edge for erasing (currently, it is $IZ-ZI$ of $Q_1^Z$), which is derived from eq.\eqref{eq:generator_inversion_symmetry}. It is a variant of the technique described in Section~\ref{section:macroscopic_code_distance}.

\begin{figure}
\centering
\includegraphics[width=0.4\textwidth]{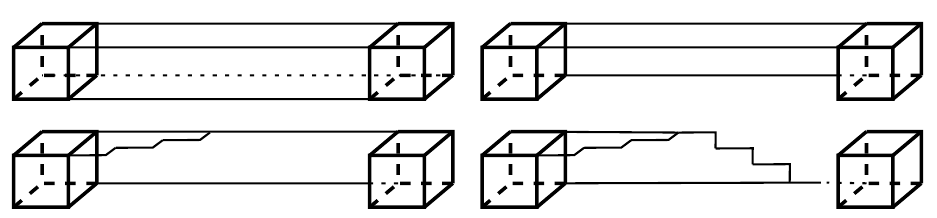}
\caption{Deformation of flat logical string segement $\zeta^{(1)}_y$. The bottom two figures are implied by `confusing constraints'. If $\zeta^{(1)}_y$ is long compared to its width, it is equivalent to a disconnected one.}
\label{fig:deformation_flat_ss_y_code_1}
\end{figure}

We shrink the support $B$ of $O$ such that now $B$ consists of the anchors plus two thin rectangles $R_z,R_x$ where $R_i$ is perpendicular to $i$-axis as shown in the top right figure of Fig.~\ref{fig:deformation_flat_ss_y_code_1}. Let us calculate what Pauli operators are possible along the edge $e$ of $R_z$ that is not contained in $R_y$. Any consecutive pair of two points on $e$ away from the anchor commute with $IZ-ZI$ and $II-IZ$ of $Q_1^Z$. Hence, the pair must be a linear combination of $XI-II$ and $IX-XI$. Since the base field is binary, there are only four combinations:
\begin{align*}
 II&-II &
 XI&-II \\
 IX&-XI &
 XX&-XI 
\end{align*}
The edge $e$ is a consistent sequence of such pairs. Going to the right (increasing $y$-coordinate), $II$ must be followed by $II$, $XI$ by $II$, $IX$ by $XI$ which must be followed by $II$, and $XX$ by $XI$ which must be followed by $II$. That is, $e$ is the identity except possibly 2 sites near the left anchor! We say such a constraint is \emph{confusing} whose solution is eventually $II$. We can repeat the argument on the next `line' $e'$ that has 1 smaller $x$-coordinate than $e$, to deduce that $e'$ is the identity except possibly 4 sites. Inductively, we conclude that $R_z$ is the identity except possibly a small `triangle' adjacent to the left anchor. See the bottom left figure of Fig.~\ref{fig:deformation_flat_ss_y_code_1}

We do the same calculation for $R_x$ with constraints given by $ IZ-ZI $ and $ II - IZ $ from the bottom face of $Q_1^Z$. The algebra is the same. Summarizing, we have shown that if $l > w + 2( 2w-1) = 5w -2$, then $\zeta^{(1)}_y$ is equivalent to a disconnected one. Note that Code 1 has three-fold rotational symmetry $\hat{x} \to \hat{y} \to \hat {z} \to \hat{x}$. Therefore, the maximum length $\phi_1(w)$ of nontrivial logical string segment of Code 1 satisfies
\[
 \phi_1(w) \leq 15w - 6 < \infty ,
\]
which completes the proof of Theorem~\ref{thm:Code1234are_free_of_strings} for Code 1.

\subsection{Inconsistent quasi-period}

Code 2 exhibits no three-fold symmetry, but instead two-fold symmetry ($\hat{x} \leftrightarrow \hat{y}$) about the plane of normal vector $(1,-1,0)$. Therefore, it suffices to consider two logical string segments $\zeta_y$ along $y$-axis and $\zeta_z$ along $z$-axis.

Let $(0,0,l_z)$ be the directional vector of $\zeta^{(2)}_z=(O_z,\Omega_{1z}, \Omega_{2z})$ pertaining to Code 2, where $O_z$ is supported on the smallest box that contains the two anchors. Using the edge $ZZ-ZI$ (read downward) of $Q_2^Z$, which is good for erasing, we further deform $O$ similar to Fig.~\ref{fig:deformation_flat_ss_y_code_1}. The support $B$ of $O_z$ is now the union of two rectangles $R_x$, $R_y$. $Q_2^Z$ acts on the edge $e$ of $R_x$ that is not contained in $R_y$ by $ZZ-ZI$ and $IZ-ZI$ (read downward). This is a confusing constraint because any two consecutive sites on $e$ must be one of
\begin{align*}
 II &- II & &
 IX - XI \\
 II &- IX & &
 IX - XX .
\end{align*}
Therefore, $e$ is the identity possibly except two sites near the bottom anchor. The same inference is applicable to $R_y$ due to the two-fold symmetry. The length $l_z$ of $\zeta^{(2)}_z$ satisfies $l_z \leq w + 2w = 3w$ if $\zeta^{(2)}_z$ is nontrivial.

Let $(0,l_y,0)$ be the directional vector of $\zeta^{(2)}_y=(O_y,\Omega_{1y},\Omega_{2y})$. Using the edge $ZZ-ZI$ of $Q_2^Z$ along $y$-axis, which is good for erasing, we may assume that $O_y$ is supported on the union of two rectangles $R_z, R_x$ as in Fig.~\ref{fig:deformation_flat_ss_y_code_2}. The constraints on the outer edges $e_1$ of $R_x$ and $e_2$ of $R_z$ ($e_2$ has bigger $x$-coordinate than $e_1$) given by $Q_2^Z$ are not confusing. We need a different argument.

\begin{figure}
\centering
\includegraphics[width=0.4\textwidth]{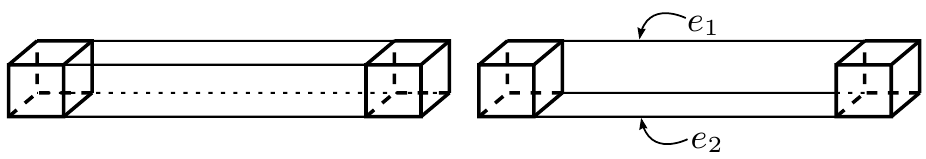}
\caption{Deformation of flat logical string segement $\zeta^{(2)}_y$.}
\label{fig:deformation_flat_ss_y_code_2}
\end{figure}

\begin{defn}
A function $f$ on the positive integers (\ie, a sequence) is \emph{eventually periodic} with \emph{quasi-period} $t \geq 1$ and \emph{offset} $n_0 \geq 0$ if $f(n+t) = f(n)$ for all $n > n_0$. The \emph{period} of $f$ is the smallest quasi-period.
\end{defn}

\noindent
If $t_1$ and $t_2$ are quasi-periods of $f$, then $f(n) = f( n + i t_1 + j t_2 )$ for sufficiently large $n$ where $ i, j $ are independent of $n$. Since there exist $i,j$ such that $i t_1 + j t_2 = t = \gcd(t_1,t_2)$, $t$ is also quasi-period. Therefore,
\begin{remark}
The period divides any quasi-period. Note also that if two sequences $f_1$ and $f_2$ have different periods $t_1, t_2$ and offsets $n_1,n_2$ respectively, there exists $n' \leq \max ( n_1, n_2 ) + \mathrm{lcm} ( t_1, t_2 )$ such that $f_1(n') \neq f_2(n')$.
\label{rem:properties_quasi_period}
\end{remark}

Supposing $\zeta^{(2)}_y$ is nontrivial and arbitrarily long, we will derive a contradiction: \emph{The operators on $e_2$ has period 3 and quasi-period a power of 2.}

The constraints on a pair of sites of $e_2$ are $ZI-IZ$ and $ZZ-ZI$. Hence, the pair is one of
\begin{align}
 II &- II & &
 XI - XX \notag \\
 IX &- XI & &
 XX - IX . \label{eq:sol_two_site_constraint1}
\end{align}
The only possible infinite sequence on $e_2$ is thus $\cdots - XI - XX - IX - XI - \cdots$, whose period is 3, or $\cdots - II - II - \cdots$. If $e_2$ is the identity, then we consider next rows of $R_z$ until we get a nontrivial row. If the entire $R_z$ is the identity, then $Q_2^Z$ imposes constraints on lower edge of $R_x$, $ZI-IZ$ and $IZ-ZZ$, whose solutions are given by eq.\eqref{eq:sol_two_site_constraint1}.
The period 3 is again revealed if nontrivial. Consider constraints on $e_1$, or more generally, on four sites of $R_x$ that form a square.
\begin{center}
\begin{tabular}{c}
$\xymatrix{
                  &  & a_{i-1,j+1} \ar@{-}[d] \ar@{-}[r] & a_{i-1,j}\\
\ar[u]^z \ar[r]^y &  & a_{i,j+1}   & a_{i,j} \ar@{-}[u] \ar@{-}[l] 
}$ \\
$ \xymatrix@!0{
 IZ \ar@{-}[d] \ar@{-}[r]& ZZ & & ZI\ar@{-}[d] \ar@{-}[r] & IZ \\
 ZI & ZI\ar@{-}[u] \ar@{-}[l]  & & ZZ & ZI\ar@{-}[u] \ar@{-}[l] 
}$
\end{tabular}
\end{center}
Here we denoted each site operator as an element of a vector-valued matrix $a$. Note that, for example, $XI$ and $IZ$ is commuting because $(10)(01)^T = 0$. Similarly, the fact that $R_x$ is commuting with $Q_2^Z$ can be expressed by a system of homogeneous equations:
\begin{align*}
 &\begin{pmatrix} 0 & 1 \\ 1 & 0 \end{pmatrix} a_{i-1,j+1}
+\begin{pmatrix} 1 & 1 \\ 0 & 1 \end{pmatrix} a_{i-1,j}  \\
+&\begin{pmatrix} 1 & 0 \\ 1 & 1 \end{pmatrix} a_{i,j+1}  
+\begin{pmatrix} 1 & 0 \\ 1 & 0 \end{pmatrix} a_{i,j}    = 0,
\end{align*}
or
\begin{equation}
 a_{i,j+1} = \begin{pmatrix} 1 & 0 \\ 0 & 0 \end{pmatrix} a_{i,j} + b( a_{i-1,j+1} , a_{i-1,j} ).
\label{eq:code-2-Rx-recursive}
\end{equation}
In order to find possible operators on $e_1$, we set $a_{0,j'}= 0$ for all $j'$. Then,
\[
 a_{1,j+1} = \begin{pmatrix} 1 & 0 \\ 0 & 0 \end{pmatrix} a_{1,j}.
\]
It is obvious that $\{ a_{1,j} \}_j$ is eventually periodic with period 1. For the rows with bigger $i$, the following holds:
\begin{lem}
Suppose a set $\{ a_{i,j} \}_{i,j}$ of vectors over the binary field $\mathbb{F}_2$ satisfies
\[
 a_{i,j+1} = M a_{i,j} + b( a_{i-1,j+1} , a_{i-1,j} ),
\]
where $M^2 = M$. If $\{ a_{i-1,j} \}_{j \geq 1}$ is eventually periodic with quasi-period $t$ and offset $n_0$, then so is $\{ a_{i,j} \}_{j \geq 1}$ with quasi-period $2t$ and offset $n_0 + t$.
\label{lem:doubling-quasi-period}
\end{lem}
\begin{proof}
Consider a sequence $\{ c_j \}_{j \geq 1}$ given by
\[
 c_{j+1} = M c_j + c' ,
\]
where $c'$ is a constant 2-component vector over $\mathbb{F}_2$ independent of $j$.
\begin{align*}
 c_{j+3}
&= M c_{j+2} + c' = M (M c_{j+1} +c') + c'\\
&= M c_{j+1} + M c' + c' = M c_j + M c' + M c' + c' \\
&= M c_j + c' = c_{j+1}
\end{align*}
Therefore, $\{ c_j \}_{j \geq 1}$ is eventually periodic with quasi-period 2 and offset 1.
Define
\[
 c_{i,j} \equiv a_{i, n_0+1+(j-1)t}
\]
where $j \geq 1$.
Clearly, $c_{i-1,j+1} = a_{i-1, n_0+1+jt} = a_{i-1,n_0+1} = c_{i-1,1}$, \ie, $\{ c_{i-1,j} \}_{j \geq 1} $ has period 1 with offset 0. Moreover, $\{ c_{i,j} \} _{i,j} $ satisfies
\[
 c_{i,j+1} = M c_{i,j} + c'_{i,j},
\]
where $c'_{i,j}=c'( a_{i-1,j'} ; n_0+1+(j-1)t \leq j' \leq n_0 + 1 + j t )$ does not depend on $j$ because $\{ a_{i-1,j} \}_j$ is eventually periodic with quasi-period $t$ and offset $n_0$. Therefore, $\{ c_{i,j} \}_j$ is eventually periodic with quasi-period 2 and offset 1. Since $a_{i,n_0+h+(j-1)t} (1 \leq h < t)$ is determined by $c_{i,j}$ and $\{ a_{i-1,n_0+h'+(j-1)t} \}_{1 \leq h' \leq h}$, we get the claim.
\end{proof}

We have shown that each row of $R_x$ that is not contained in $R_z$ is eventually periodic, and as we decrease by 1 the $z$-coordinate of a row of $R_x$ we get its quasi-period doubled. Since the initial row has period 1, the quasi-period of each row is some power of 2. We need to show the same thing continues to hold as we move toward $e_2$. The constraints pertaining to the last edge of $R_x$, \ie, the intersection of $R_x$ and $R_z$, are
\[
 \xymatrix@!0{
 IZ \ar@{-}[d] \ar@{-}[r]& ZZ & & II\ar@{-}[d] \ar@{-}[r] & II \\
 ZI & ZI\ar@{-}[u] \ar@{-}[l]  & & IZ & ZZ . \ar@{-}[u] \ar@{-}[l] 
}
\]
The recursive equation is then
 \begin{equation}
 a_{w,j+1} = \begin{pmatrix} 1 & 0 \\ 1 & 1 \end{pmatrix} a_{w,j} + b( a_{w-1,j+1} , a_{w-1,j} ).
\label{eq:code-2-Rx-Rz-recursive}
\end{equation}
One can repeat the proof of Lemma~\ref{lem:doubling-quasi-period}, except that now one has to show that a sequence defined by
\[
 f_{j+1} = N f_j + f' ,
\]
where $N^2 = I$ and $f'$ is a constant vector, is always eventually periodic with quasi-period some power of 2. This is easy:
\begin{align*}
 f_{j+2} 
&= N ( N f_j + f' ) + f' = f_j + f'' \\
 f_{j+4}
&= f_{j+2} + f'' = f_j + f'' + f'' = f_j.
\end{align*}
\noindent
We have proved
\begin{lem}
Suppose a set $\{ a_{i,j} \}_{i,j}$ of vectors over the binary field $\mathbb{F}_2$ satisfies
\[
  a_{i,j+1} = N a_{i,j} + b( a_{i-1,j+1} , a_{i-1,j} ).
\]
where $N^2 = I$. If $\{ a_{i-1,j} \}_{j \geq 1}$ is eventually periodic with quasi-period $t$ and offset $n_0$, then so is $\{ a_{i,j} \}_{j \geq 1}$ with quasi-period $4t$ and offset $n_0$.
\label{lem:double-doubled-quasi-period}
\end{lem}

The exact recursive equation for the next row is slightly more complicated since there are stabilizer generators meeting three rows, $(w-1), w, (w+1)$-th. However, we do not need detailed information how $Q_2^Z$ acts on $(w-1)$- and $w$-th rows to infer the quasi-period of $(w+1)$-th. The constraints on $i$-th row ($i > w$) are as the following.
\begin{center}
\begin{tabular}{c}
$\xymatrix{
\ar[r]^y \ar[d]^x &  & a_{i-1,j+1} \ar@{-}[d] \ar@{-}[r] & a_{i-1,j}\\
                  &  & a_{i,j+1}   & a_{i,j} \ar@{-}[u] \ar@{-}[l] 
}$ \\
$ \xymatrix@!0{
 ZI \ar@{-}[d] \ar@{-}[r]& IZ & & ZZ\ar@{-}[d] \ar@{-}[r] & ZI \\
 IZ & ZZ\ar@{-}[u] \ar@{-}[l]  & & ZI & ZI . \ar@{-}[u] \ar@{-}[l] 
}$
\end{tabular}
\end{center}
The recursive equation is then
\[
 a_{i,j+1} = \begin{pmatrix} 1 & 0 \\ 1 & 1 \end{pmatrix} a_{i,j} 
   + b\begin{pmatrix} a_{i-1,j+1} ,& a_{i-1,j},\\ a_{i-2,j+1} ,& a_{i-2,j} \end{pmatrix},
\]
where $i > w$. We may regard that both $(i-2)$-th and $(i-1)$-th rows have the same quasi-period and the same offset. Hence, we can apply Lemma~\ref{lem:double-doubled-quasi-period} to the rows in $R_z$.

We have reached $e_2$ starting from $e_1$. Since $e_1$ has period 1, and each following row has quasi-period some power of 2, $e_2$ must have quasi-period $2^{w-2} \times 4^{w} = 2^{3w-2}$ with offset $2^{w-2}-1$. The true period 3 of $e_2$ must divide $2^{3w-2}$, which is a contradiction. By Remark~\ref{rem:properties_quasi_period}, we conclude that the length of a nontrivial logical string segment $\zeta^{(2)}_y$ is $ \leq w + (2^{w-2}-1) + 3 \cdot 2^{3w-2}$. Since the length of $\zeta^{(2)}_z$ is $\leq 3w$, the maximum length $\phi_2(w)$ of nontrivial logical string segments of Code 2 satisfies
\[
\phi_2(w) \leq 5w + 2^{3w+1} < \infty,
\]
which completes the proof of Theorem~\ref{thm:Code1234are_free_of_strings} for Code 2.

The proof of Theorem~\ref{thm:Code1234are_free_of_strings} for Code 3 and 4 using similar technique is given in Appendix~\ref{section:pf_thm_code_34}.

\section{Number of encoded qubits and logical operators in finite periodic lattices}
\label{section:k_vs_L}

\subsection{Number of encoded qubits}

The number of encoded qubits of a stabilizer code in any finite periodic lattice will be obtained, once we know all the algebraic relations of stabilizer generators in the infinite lattice. A nontrivial example of this approach is given in \cite{BravyiLeemhuisTerhal2010Topological} for Chamon model. It has a nice property that the product of stabilizer generators becomes the identity only when they form body-diagonal surfaces in the infinite lattice. Since there are 4 body-diagonals, the number of encoded qubits of Chamon model in the periodic lattice $\mathbb{Z}_{2 p_x} \times \mathbb{Z}_{2 p_y} \times \mathbb{Z}_{2 p_z}$ is $k = 4 \gcd(p_x,p_y,p_z)$.

Our cubic codes exhibit even more peculiar dependencies of $k$ on the linear size of the periodic lattice $\mathbb{Z}_L^3$. We found empirical formulae for $k=k(L)$ of Code 0, 1, 2, 3, 4, which are exact if $ 2 \leq L \leq 200 $. For ease of notation, we define the \emph{divisibility function} $q_n$ on positive integers for each positive integer $n$ by
\[
 q_n(L) =  
  \begin{cases}
   1 & \text{if $n$ divides $L$}, \\
   0 & \text{otherwise}.
  \end{cases}
\]
The formulae for $k$ are given in Table~\ref{tb:k_vs_L}.

\begin{table}
\centering
\begin{tabular}{c|c}
\hline
Code & $k(L)$ \\
\hline
0 & $L + 3 \cdot 2^r \left( q_2 + 2 q_7 + 8 q_9 + 48 q_{63} + 64 q_{65} + 18 q_{171} \right)$ \\
1 & $2 \left[ 1 - 2 q_2 + 2^{r+1} \left( q_2 + 12 q_{15} + 60 q_{63} \right) \right]$ \\
2 & $2^{r+1} \left( 1 + 6 q_7 + 6 q_{21} + 30 q_{31} + 60 q_{63} + 126 q_{127} \right)$ \\
3 & $\begin{matrix} 2^{r+1} ( 1 + 8 q_{15} + 6 q_{21} + 40 q_{31} + 42 q_{63} \\ + 16 q_{85} + 112 q_{127} )\end{matrix}$ \\
4 & $2^{r+1} \left( 1 + 2 q_3 + 8 q_{15} + 40 q_{31} + 48 q_{63} + 112 q_{127} \right)$ \\
\hline
\end{tabular}
\caption{Exact empirical formulae for the number $k$ of encoded qubits in periodic finite lattice $\mathbb{Z}_L^3$ as a function of $L$ ($2 \leq L \leq 200$). Here, $q_n=q_n(L)$ is the divisibility function, and $r=r(L)$ is the largest integer such that $2^r$ divides $L$.}
\label{tb:k_vs_L}
\end{table}

\begin{table}
\centering
\begin{tabular}{c|c|c}
\hline
Code & Lower bound & Upper bound \\
\hline
0 & $L + 6 q_2 $ & $12L -12$ ($4L$ if $ 7 \nmid L $) \\
1,2,3,4 & 2 & 4L \\
\hline
\end{tabular}
\caption{Lower and Upper bound on the number of encoded qubits in periodic finite lattice $\mathbb{Z}_L^3$. Here, $q_n=q_n(L)$ is the divisibility function.}
\label{tb:lower_upper_bound_k}
\end{table}

Although we only know empirical formulae of $k$ for small values of $L$, we can prove lower and upper bounds on $k$. Since there are equal number of stabilizer generators and physical qubits, $k$ is equal to the number of independent algebraic relations of generators. It is obvious from the definition of cubic codes that $k \geq 2$ since the product of all generators is the identity in any periodic lattice. In addition, $k$ is always an even number for CSS cubic codes because of the duality between $X$- and $Z$-type generators.

The non-CSS Code 0 is more complicated than the CSS codes. We prove the bounds on $k$ for Code 0 in Appendix~\ref{section:code0}. Consider Code 1 in $\mathbb{Z}_L^3$. We find an independent set of generators of $Z$-type to derive the upper bound. Let $S$ be the set of $Z$-type generators lying outside a straight tunnel $T$ of length $L$ parallel to $z$-axis, whose cross-section is an $1 \times 2$ rectangle enclosed by 6 sites. There are $L^3 - 2L$ generators in $S$. We claim that $S$ is an independent set of generators, and hence $k \leq 4 L$.

Suppose a linear combination $O$ of generators in $S$ is the identity operator. We show $O$ is the zero combination. Choose the origin of the coordinate system such that the sites enclosing the cross-sectional rectangle of $T$ are described by $ x = \pm 1, ~ y=0,1 $. Let $l^{(1)}$ be the straight line given by $ x=0, y=1 $. See Fig.~\ref{fig:tunnel}. Since $O$ is the identity, in particular, $l^{(1)}$ is acted on by the identity. Every unit edge $e_i$ in $l^{(1)}$ connecting $(0,1,i+1)$ and $(0,1,i)$ is one of (read downward)
\begin{align}
 II &- II & &
 ZI - IZ  \notag \\
 IZ &- II & &
 ZZ - IZ \label{eq:action_generator_unit_edge_code_1}
\end{align}
which is canceled by the neighboring unit edges. If $e_1 = II-II$, then $e_2= II-II$. If $e_1 = ZI-IZ$, $e_2 = IZ-II$. If $e_1 = IZ-II$, $e_2=II-II$. If $e_1 = ZZ-IZ$, $e_2 = IZ - II$. We see that $\{ e_i \}_{i \geq 1}$ is eventually $II-II$. Since $l^{(1)}$ is periodic, $II-II = e_{L+1} = e_1$. We conclude that the coefficients of generators in $S$ around $l^{(1)}$ are all zero. Now $O$ is a linear combination of generators lying outside the enlarged tunnel whose cross-sectional rectangle is described by $x = \pm 1, ~ y = 0,2 $.

\begin{figure}
\centering
\includegraphics[width=0.4\textwidth]{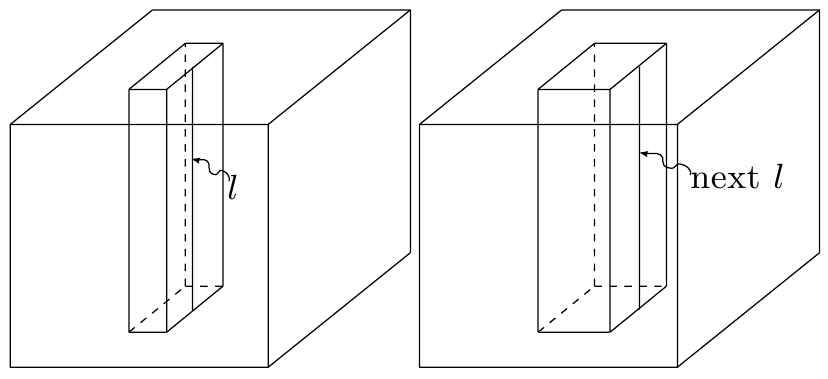}
\caption{Tunnel used to derive the upper bound on $k$, the number of encoded qubits. The stabilizer generators outside the tunnel are independent.}
\label{fig:tunnel}
\end{figure}

We can repeat the argument inductively on the lines parallel to $z$-axis, each of which is given by
\[
\{ (0,y',z) \in \Lambda ~|~ 0 \le z < L \} ,
\]
where $y'=2,\ldots,L-1$, until the tunnel becomes a slab of width 2. (The slab is in between two planes $ x = \pm 1 $.) Consider a straight line $l'^{(1)}$ parallel to $z$-axis given by $ y=0, x=1$. Since possible operators on each unit edge of $l'^{(1)}$ are again given by eq.\eqref{eq:action_generator_unit_edge_code_1}, we argue similarly to conclude all coefficients in $O$ are zero. This completes the proof of the upper bound on $k$ of Code 1.

For Code 2, we use the same initial tunnel $T$. The initial $l^{(2)}$ is chosen to be the line parallel to $z$-axis given by $x=0, ~ y=0$. For Code 3, we choose $l^{(3)} = l^{(1)}$.
Every unit edge of $l^{(3)}$ is one of
\begin{align*}
 II &- II & &
 IZ - II  \\
 IZ &- ZZ & &
 II - ZZ
\end{align*}
which must be canceled by neighboring unit edges. The only choice is $II-II$. Due to two-fold symmetry of Code 2 and 3, the rest of calculation is easy and proves the upper bound. For Code 4, we choose the same initial tunnel $T$ and the initial line $l^{(4)} = l^{(1)}$. Arguing as above, one enlarges the tunnel until it becomes a slab of width 2. Now consider any horizontal line $l'$ ($z=z'$) in $x=1$ plane. We see that every edge is one of
\begin{align*}
 II &- II & &
 ZI - IZ  \notag \\
 II &- IZ & &
 ZI - II
\end{align*}
which must be canceled by neighboring edges. The only choice is $II-II$. This completes the proof for the upper bound on $k$ of Code 4.

\subsection{Plane logical operators}

Theorem~\ref{thm:Code1234are_free_of_strings} says, in particular, that Code 1, 2, 3, 4 do not have any nontrivial logical operators that is supported on a thin strand $\{1,\ldots,w\} \times \{1,\ldots,w\} \times \mathbb{Z}_L$ where $ w = O(1) $. Bravyi and Terhal~\cite{BravyiTerhal2008no-go} showed that any local stabilizer code must have nontrivial logical operator supported on a thin \emph{slab} $\{1,\ldots,w\} \times \mathbb{Z}_L^2$, where $w$ is the interaction range. Indeed, cubic codes have logical operators on $(w=1)$-slab, \ie, \emph{plane logical operators}. We consider the simplest plane logical operators that are repetition of a single site operator (two qubit operator).
For a single site (two qubit) operator $E$, we define
\[
 \sigma^{[a,b,c]}_{E} = \bigotimes_{(a,b,c)\text{-plane}} E 
\]
to be the tensor product of $E$ over the plane orthogonal to $(a,b,c)$.
$\sigma_E^{\hat{x}}$ is logical if and only if
\begin{align*}
 \lambda(E,A_Z) + \lambda(E,B_Z) + \lambda(E,C_Z) + \lambda(E,D_Z) &= 0 ,\\
 \lambda(E,A'_Z) + \lambda(E,B'_Z) + \lambda(E,C'_Z) + \lambda(E,D'_Z) &= 0 .
\end{align*}
which is equivalent to $ \lambda( E, A_Z B_Z C_Z D_Z ) = 0 $ and $ \lambda( E, A'_Z B'_Z C'_Z D'_Z ) = 0 $.
(The two are in fact equivalent because one of our conditions defining cubic codes requires that the product of all eight corner operators be $II$.)
We see that $\sigma^{[100]}_{IX}$ is a logical operator of Code 1, and so is $\sigma^{[1,-1,0]}_{ZZ}$. Moreover, if the linear lattice size $L$ is odd, they anti-commute, and hence are both nontrivial. In this way, one can easily find logical operators of form $\sigma^{[abc]}_{E}$. We make it clear that the plain logical operators found in this way do not generate all the logical operators. Empirically, for some special lattice size, \emph{e.g.} $L=8,15,63$, there are many more logical operators that cannot be described by the plane logical operators.

\section{Thermal stability}
\label{section:thermal_stability}

In this section, we discuss the energetics of implementing nontrivial logical operators of Code 1. A natural choice of Hamiltonian for a local stabilizer code is the sum of local generators:
\begin{equation}
 H = - \frac{1}{2} \sum_{p \in \Lambda} (Q_1^Z)_p +  (Q_1^X)_p,
\label{eq:memory_hamiltonian}
\end{equation}
whose ground space is the code space. We consider adverse logical operations on the code space by thermal environment. Quantum tunneling between two different ground states is exponentially suppressed in the system size since eq.\eqref{eq:memory_hamiltonian} is gapped and the code distance is at least $L$. We may model the thermal noise as a Markovian chain of actions to the system by Pauli operators of weight one. At each step, a Pauli operator by the environment will excite some terms in the Hamiltonian. Such excitation is completely determined by the accumulated \emph{syndrome}, the data that describes which stabilizers are unsatisfied. Formally, a syndrome can be viewed as a $\mathbb{Z}_2$-valued function on the stabilizer group.

We claim that the energy of any partial implementation of $\sigma_{IX}^{[100]}$ is proportional to its boundary length. A partial implementation of a logical operator is a Pauli operator that is a restriction of the logical operator on a subset of sites. It has excitations, if any, only along the boundary. Note that a syndrome caused by a Pauli operator is the sum of syndromes caused by Pauli operators of weight one, each of which is expressed as a cube in the dual lattice as shown in Fig.~\ref{fig:syndrome_cube}. The syndrome corresponding to the partial implementation of $\sigma_{IX}^{[100]}$ is expressed as a stack of $IX$-cubes. An outer boundary point of the stack is either a vertex of one syndrome cube, a point where two cubes meet, or a point where three cubes meet. In any case, one can easily verify that there exists a filled dot \ie, an excitation within distance one from the boundary point. Therefore, the number of excitations (the energy) of any partial implementation of $\sigma_{IX}^{[100]}$ is at least the length of the outer boundary. One checks that the same holds for plane logical operators $\sigma_{XI}^{[110]}$ and $\sigma_{XX}^{[1,-1,0]}$.

\begin{figure}
\centering
\includegraphics[width=0.3\textwidth]{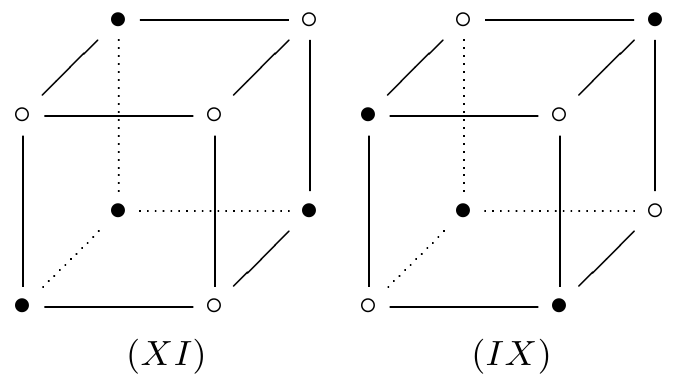}
\caption{Syndrome cubes of Code 1 caused by a Pauli operator of weight one. Filled dots indicate excited $Q_1^Z$. See also Fig.~\ref{fig:generator_Code_01234}}
\label{fig:syndrome_cube}
\end{figure}

This property that energy is proportional to the boundary length of partial logical operator mimics that of 4D toric code. Still, it is needed to show that any partial implementation of \emph{any} nontrivial logical operator confronts high energy barrier. But we do not have a complete description of all equivalent variants of, say, $\sigma_{IX}^{[100]}$ of Code 1. Moreover, we do not have a list of logical operators for arbitrary system size.

We finally note that the system governed by eq.\eqref{eq:memory_hamiltonian} is a quantum glass~\cite{Chamon2005Quantum}; any isolated excitation (defect) cannot propagate easily to meet its partner to annihilate if the temperature is low. The situation is more stringent than Chamon model where one can move even number of excitations by string operators. In Code 1, there is no local Pauli operator that can move a localized, but not locally created, set of excitations to nearby location congruently.

\section{Discussion}
\label{section:discussion}

We have defined logical string segments for local stabilizer codes, whose length determines the existence of nontrivial string logical operator. We classified translation-invariant CSS stabilizer codes in three dimensions with two stabilizer generators, and also found the unique non-CSS stabilizer code with the special symmetry. We showed that some of the CSS cubic codes do not have string logical operators. The codes without string logical operator exhibit peculiar dependence of the number of logical qubits on the system size. This is because there are complicated algebraic relations of stabilizer generators, for which we do not have full description.

The thermodynamic stability of encoded information in the code is still an open problem. It has been proved that in order for a quantum memory based on stabilizer codes to have a thermal stability, it suffices to have a good error correcting procedure by which every low energy syndromes are \emph{good}~\cite{ChesiLossBravyiEtAl2009Thermodynamic}. Returning to Code 1 with the system size such that $k=2$, there is only one algebraic relation of stabilizer generators: The product of all generators is equal to the identity. Therefore, there is the unique constraint (conserved charge) on the space of syndromes: The number of excitations must always be even. In other words, any state with two excitations arbitrarily far apart is allowed. Since there is no string operator, a Pauli operator that causes such excitations must be of complex shape. A good error correcting procedure should answer how such a syndrome is formed. On the other hand, if every three dimensional local stabilizer code is not stable, our codes illustrate that one should not attempt to prove it by showing the existence of string logical operator.

It has been recently proved that any Pauli walk that results in a nontrivial logical operator of Code 1 must experience energy barrier of height $\Omega( \log L )$ \cite{BravyiHaah2011}.

\begin{acknowledgments}
This research was supported in part
by NSF under Grant No. PHY-0803371, 
by ARO Grant No. W911NF-09-1-0442,
by DOE Grant No. DE-FG03-92-ER40701,
and by Korea Foundation for Advanced Studies.
I thank Salman Beigi, Liang Jiang, and especially John Preskill for useful discussions.
\end{acknowledgments}

\appendix

\section{Proof of Theorem~\ref{thm:Code1234are_free_of_strings} for Code 3, 4}
\label{section:pf_thm_code_34}

For any logical string segment $\zeta$, we first deform it such that it is supported on the smallest box containing two anchors. This is easily done if there are two orthogonal edges that are good for erasing on each face of $Q^Z_{3,4}$. However, Code 3 and 4 have faces that do not have such two edges. Let $II$ of $Q_{3}^Z$ be at $(0,0,0)$. The $(x=0)$-face of $Q_3^Z$ has only one edge $IZ-ZZ$ that is good for erasing. Using this, we can shrink the $x$-size of the support of $\zeta$ from positive $x$-axis, except a $y$-plane $R_y$ of thickness one. The constraints on the outer edge $e$ of $R_y$ are $ IZ - II, IZ-ZZ$. Every unit edge in $e$ is thus either $II-II, XI-II, II-XX, XI-XX$. Since $e$ is finite, every site of $e$ is $II$. The $(y=0)$-face of $Q_3^Z$ is symmetric to $(x=0)$-face. Code 4 has two faces that do not have two orthogonal good edges for erasing: $(x=0)$- and $(z=0)$-face of $Q_4^Z$. One can repeat the same argument.

\begin{figure}
\centering
\includegraphics[width=0.45\textwidth]{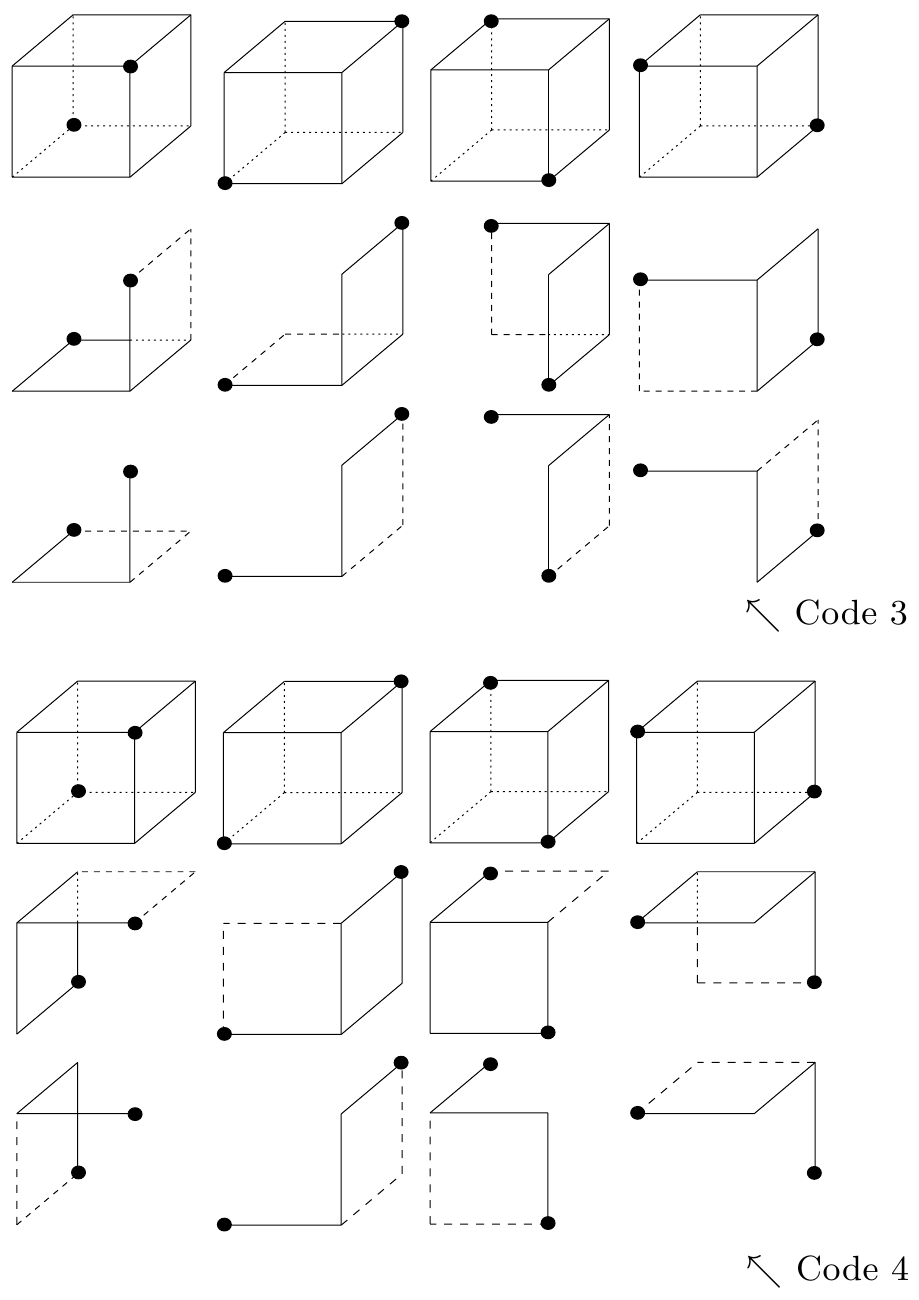}
\caption{Deformations of logical string segments of Code 3, 4. Anchors are marked by the filled dots.}
\label{fig:reduction_flat_ss_code_34}
\end{figure}

Fig.~\ref{fig:reduction_flat_ss_code_34} shows that every logical string segment of Code 3 and 4 is equivalent to a union of logical string segments whose directional vectors are parallel to the coordinate axes. Code 3 has two-fold symmetry $\hat{x} \leftrightarrow \hat{y}$, and therefore it is enough to consider $\zeta^{(3)}_z$ along $z$-axis and $\zeta^{(3)}_y$ along $y$-axis. $\zeta^{(3)}_z$ can be further deformed by $IZ-ZZ$(read downward) of $Q_3^Z$. See Fig.~\ref{fig:deformation_flat_ss_zy_code_3}. We consider the constraints $IZ-II, IZ-ZZ$ given by $Q_3^Z$ on the external edge. They are confusing because we have to obtain a consistent solution on the edge using
\begin{align*}
 II &- II &  XI &- II \\
 II &- XX &  XI &- XX .
\end{align*}
Therefore, a nontrivial $\zeta^{(3)}_z$ has length $\leq w+(2w-1)$. $\zeta^{(3)}_y$ is deformed using $ZZ-ZI$(read rightward) of $Q_3^Z$ as depicted in Fig.~\ref{fig:deformation_flat_ss_zy_code_3}. One edge is confused, and we are left with a rectangle normal to $z$-axis. $e_1$ is constrained by $IZ-IZ$ and $II-ZZ$. A possible neighboring pair is one of $II-II$, $XI-II$, $IX-XX$, and $XX-XX$. Thus, any nontrivial $e_1$ is eventually $\cdots - XX - \cdots$ of period 1. $e_2$ is constrained by $IZ-ZZ$ and $ZZ-ZI$. A possible neighboring pari is one of $II-II$, $XI-XX$, $IX-XI$, and $XX-IX$. A nontrivial $e_2$ is $\cdots-XI-XX-IX-XI-\cdots$ of period 3. In between $e_1$ and $e_2$, we have
\begin{center}
\begin{tabular}{c}
$\xymatrix{
\ar[r]^y \ar[d]^x &  & a_{i,j} \ar@{-}[d] \ar@{-}[r] & a_{i,j+1}\\
                  &  & a_{i-1,j}   & a_{i-1,j+1} \ar@{-}[u] \ar@{-}[l] 
}$ \\
$ \xymatrix@!0{
 II \ar@{-}[d] \ar@{-}[r]& ZZ & & IZ\ar@{-}[d] \ar@{-}[r] & IZ \\
 ZZ & ZI\ar@{-}[u] \ar@{-}[l]  & & IZ & ZZ .\ar@{-}[u] \ar@{-}[l] 
}$
\end{tabular}
\end{center}
The recursive equation is then
\[
 a_{i,j+1} = \begin{pmatrix} 0 & 1 \\ 0 & 1 \end{pmatrix} a_{i,j} + b( a_{i-1,j+1}, a_{i-1,j} ) .
\]
Since $M = \begin{pmatrix} 0 & 1 \\ 0 & 1 \end{pmatrix} = M^2$, the quasi-period of $i$-th row is twice as that of $(i-1)$-th by Lemma~\ref{lem:doubling-quasi-period}. Therefore, a nontrivial $\zeta^{(3)}_y$ has finite length. See Remark~\ref{rem:properties_quasi_period}. This completes the proof of Theorem~\ref{thm:Code1234are_free_of_strings} for Code 3.

\begin{figure}
\centering
\includegraphics[width=0.4\textwidth]{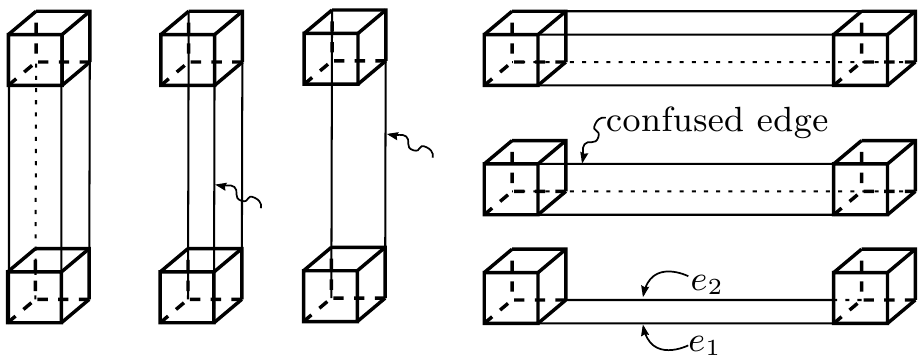}
\caption{Deformation of flat logical string segments $\zeta^{(3)}_z$ and $\zeta^{(3)}_y$.}
\label{fig:deformation_flat_ss_zy_code_3}
\end{figure}

Code 4 has no symmetry, and we need to check $\zeta^{(4)}_x, \zeta^{(4)}_y, \zeta^{(4)}_z$.
The calculation is straightforward following Fig.~\ref{fig:deformation_flat_ss_code_4}. $\zeta^{(4)}_z$ must be short if nontrivial, because
\begin{align*}
 \{ ZI-II, IZ-ZI \} \\
 \{ IZ-ZI, ZZ-ZI \}
\end{align*}
(read downward) are confusing constraints. $\zeta^{(4)}_x$ must be short because
\begin{align*}
 \{ ZI - II, ZI-IZ \} \\
 \{ ZI - II, IZ-ZI \}
\end{align*}
(read as decreasing $x$-coordinate) are confusing constraints. $\zeta^{(4)}_y$ is deformed to have eventually periodic two edges $e_1,e_2$. On $e_1$, the constraint is $\{ II-IZ, ZI-ZI \}$ whose solution is a linear combination of $IX-II$ and $ XI-XI$. Hence, $e_1$ is of period 1. On $e_2$, the constraint is $\{ IZ-ZZ, ZI-IZ \}$ whose solution is a linear combination of $IX-XI$ and $XI-XX$. Hence, $e_2$ is of period 3. In between $e_1$ and $e_2$, we have:
\begin{center}
\begin{tabular}{c}
$\xymatrix{
                  &  & a_{i-1,j} \ar@{-}[d] \ar@{-}[r] & a_{i-1,j+1}\\
\ar[u]^z \ar[r]^y &  & a_{i,j}   & a_{i,j+1} \ar@{-}[u] \ar@{-}[l] 
}$ \\
$ \xymatrix@!0{
 ZI \ar@{-}[d] \ar@{-}[r]& IZ & & IZ\ar@{-}[d] \ar@{-}[r] & ZZ \\
 II & IZ\ar@{-}[u] \ar@{-}[l]  & & ZI & ZI\ar@{-}[u] \ar@{-}[l]
}$
\end{tabular}
\end{center}
The recursive equation is
\[
 a_{i,j+1} = \begin{pmatrix} 1 & 0 \\ 0 & 0 \end{pmatrix} a_{i,j} + b( a_{i-1,j} , a_{i-1,j+1} ),
\]
to which we apply Lemma~\ref{lem:doubling-quasi-period}. We completed the proof of Theorem~\ref{thm:Code1234are_free_of_strings} for Code 4.

\begin{figure}
\centering
\includegraphics[width=0.4\textwidth]{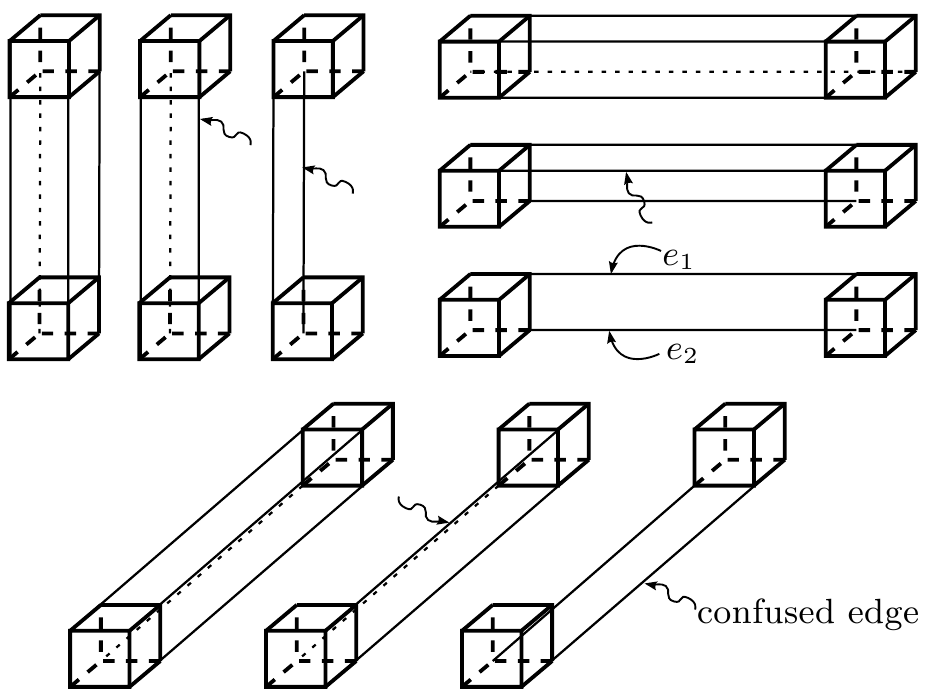}
\caption{Deformation of flat logical string segments $\zeta^{(4)}_z$, $\zeta^{(4)}_y$ and $\zeta^{(4)}_x$.}
\label{fig:deformation_flat_ss_code_4}
\end{figure}

\section{Cubic codes with string logical operators}
\label{section:codes_w_string_logical_operators}

For a local stabilizer code defined on a periodic finite lattice $\mathbb{Z}_L^3$, suppose there exists a nontrivial logical operator $O$ supported on a strand $\{1,\ldots,w\} \times \{1,\ldots,w\} \times \mathbb{Z}_L$. If $w+2(r-1) < L$ where $r$ is the interaction range of the stabilizer generators, then we may lift $O$ to a periodic logical operator $O'$ of the code on the infinite lattice supported on $\{1,\ldots,w\} \times \{1,\ldots,w\} \times \mathbb{Z}$. Suppose further that any finite logical operator is trivial in the infinite lattice.
If a contiguous part $\zeta$ of $O'$ is trivial as a logical string segment, then every congruent part of $\zeta$ is also trivial; $O'$ is a disconnected union of finite logical operators. Hence $O$ is trivial, a contradiction. Therefore, any contiguous part of $O'$ is a nontrivial logical string segment. The maximum length of the nontrivial logical string segment is infinite.

In Table~\ref{tb:string_logical_operators} we list nontrivial logical operators supported on strands for the cubic codes marked with $\dagger$ in Table~\ref{tb:list_codes}. Recall that $E[v]_p$ denotes the Pauli operator $\cdots \otimes E \otimes E \otimes \cdots$ on the line along the vector $v$ passing through $p$. For example, $ZI[0,0,1]_{(3,0,0)} ZZ[0,0,1]_{(2,0,0)}$ represents (when $L=5$)
\[
 \begin{matrix}
II & II & ZI & ZZ & II \\
II & II & ZI & ZZ & II \\
II & II & ZI & ZZ & II \\
II & II & ZI & ZZ & II \\
II & II & ZI & ZZ & II .
\end{matrix}
\xymatrix@!0{
&  & \\
&  & \ar[l]^x \ar[u]^z
}
\]

\begin{table}
\begin{tabular}{c|c|c}
\hline
Code & string logical operator & complement ($L=5$) \\
\hline
11 & $ZZ[\hat{z}]_{(000)} ZI[\hat{z}]_{(100)}$ & $XI[\hat{y}]_{(000)} IX[\hat{y}]_{(100)}$ \\
12 & $IZ[\hat{z}]_{(000)} ZI[\hat{z}]_{(010)}$ & $XI[\hat{x}]_{(000)} XX[\hat{x}]_{(010)}$ \\
13 & $ZZ[\hat{z}]_{(000)} IZ[\hat{z}]_{(010)}$ & $\sigma^{[010]}_{IX}(000)$ \\
14 & $IX[\hat{z}]_{(000)} XI[\hat{z}]_{(010)}$ & $\sigma^{[001]}_{IZ}(000)$ \\
15 & $ZI[\hat{y}]_{(000)} ZZ[\hat{y}]_{(100)}$ & $IX[\hat{z}]_{(000)} XI[\hat{z}]_{(100)}$ \\
16 & $ZZ[101]_{(000)} IZ[101]_{(100)}$         & $IX[110]_{(000)} XI[110]_{(100)}$ \\
17 & $ZZ[\hat{x}]_{(000)} IZ[\hat{x}]_{(001)}$ & $ \sigma^{[001]}_{IX}(000)$ \\
\hline
\end{tabular}
\caption{Nontrivial string logical operators. $E[v]_p$ denotes the Pauli operator such that $E$ is repeated along $v$ passing through $p$, and $\sigma^v_E(p)$ denotes the Pauli operator such that $E$ is repeated on the plane containing $p$ perpendicular to $v$.}
\label{tb:string_logical_operators}
\end{table}

\section{Code 0}
\label{section:code0}
\subsection{Three-fold symmetry}
We remark that there is a three-fold symmetry for the generators of Code 0. (Fig.~\ref{fig:generator_Code_01234}) If we rotate $Q_0$ by $120^\circ$ about $(1,1,1)$-axis, and then apply the transformation $?X \to ?Y \to ?Z \to ?X$, we see that $Q_0$ is invariant. This is in fact expected from the commutation relation $\omega$ of corner operators of $Q_0$; we calculated $Q_0$ from $\omega$. The trasformation is symplectic. If we order the basis for $\mathcal{P}_2$ as $\{ XX,XZ,ZX,ZZ \}$, the transformation is $S R_{120^\circ} $, where $R_{120^\circ}$ is the rotation about $(1,1,1)$-axis and
\[
S=
 \begin{pmatrix}
  1 & 0 & 0 & 0 \\ 
  0 & 1 & 0 & 0 \\
  0 & 0 & 1 & 1 \\
  0 & 0 & 1 & 0
 \end{pmatrix}.
\]
One directly checks that $S$ preserves the symplectic form of the abelianized $\mathcal{P}_2$
\[
\lambda=
 \begin{pmatrix}
  0 & 1 & 0 & 0 \\ 
  1 & 0 & 0 & 0 \\
  0 & 0 & 0 & 1 \\
  0 & 0 & 1 & 0
 \end{pmatrix},
\]
\ie, $S^T \lambda S = \lambda$. This is an example of Lemma~\ref{lem:unique_realization_of_commutation_relation}. A direct consequence of the three-fold symmetry is that logical operators always appear as a triple.

\subsection{Bounds on the number of encoded qubits}

We show $k$, the number of encoded qubits of Code 0 defined on the periodic finite lattice $\mathbb{Z}_L^3$, satisfies $k \geq L$. To this end, we present an algebraic relation of stabilizer generators in the infinite lattice that can be embedded into a periodic lattice of arbitrary linear size. An algebraic relation of stabilizer generators is a distribution of generators whose product is the identity.
(Formally, the space of algebraic relations is the kernel of the linear map $\{ f:\Lambda^* \to G \} \to \{ \sigma : \Lambda \to \mathcal{P}_2 \}$ where $\Lambda^*$ is the dual lattice to $\Lambda = \mathbb{Z}^3$, $G$ is the abelian group of \emph{labels} of stabilizer generators, and $\mathcal{P}_2$ is the abelianized Pauli group on a single site. $G$ is isomorphic to $\mathbb{Z}_2 \times \mathbb{Z}_2$ for our cubic codes because there are two types of generators. This map is meaningful for translation-invariant stabilizer codes.)
Note that the set of the locations of generators in an algebraic relation need not be finite; the product is well-defined if the number of generators acting on each site is finite, \ie, locally finite. For some relations the set of the locations of generators form a sublattice, which we call \emph{relation lattice}, and can be described by a unit cell and basis vectors. Indeed, the relations that gives the lower bound $k \geq L$ forms sublattices.

The relation we consider first has its own lattice structure $R_1$ with a basis $ \{ (1,0,-1),(0,1,-1) \}$, the unit cell of which consists of $Q_0$ and $Q_0^P$ such that $Q_0^P$ is at $(0,0,1)$ relative to $Q_0$. In other words, Fig.~\ref{fig:relation_lattice_code_0} is repeated according to $(1,0,-1),(0,1,-1)$. Since the two basis vectors have period $L$ in the finite lattice of linear size $L$, one can embed this relation into any finite periodic lattice, and there are $L$ linearly independent such embeddings via translations along $(1,0,0)$. Therefore, we have at least $L$ independent algebraic relations of stabilizer generators for a $L \times L \times L$ lattice. This proves that $ k \geq L $. Note that this relation lattice is invariant under the three-fold symmetry of Code 0. If we consider the embeddings into finite periodic lattice with arbitrary three linear sizes $\mathbb{Z}_{L_x} \times \mathbb{Z}_{L_y} \times \mathbb{Z}_{L_z}$, we see that $ k \geq \gcd( L_x, L_y, L_z ) $.

\begin{figure}
\centering
\includegraphics[width=0.4\textwidth]{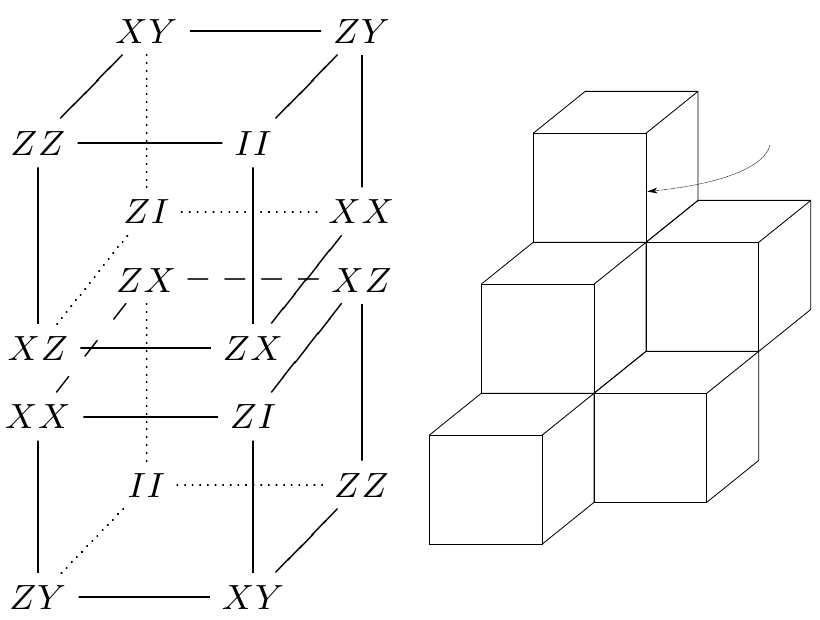}
\caption{Unit cell of a relation lattice of Code 0. The relation lattice has basis $\{ (-1,0,1),(0,-1,1) \}$. The boxes displays the configuration of $Q_0$'s in the relation. $Q_0^P$ lies on the top of $Q_0$. The line designated by the arrow is acted on trivially by four unit cells around it, as one can directly check using the Pauli operator diagrams.}
\label{fig:relation_lattice_code_0}
\end{figure}

Another relation lattice $R_2$ has basis $\{ (2,0,0),(0,2,0),(0,0,2) \}$ with unit cell such that $Q_0^P$ is at $(1,1,1)$ relative to $Q_0$. $R_2$ is embedded into any finite lattice of even linear size. There are 8 independent relations whose underlying lattice structure is $R_2$. However, $R_1$ and $R_2$ are not all independent. It is not hard to see $R_2$ gives 6 more algebraic relations. This proves the lower bound $k \geq L + 6 q_2$.

It should be pointed out that $-I \notin \mathcal{S}$ is not automatically guaranteed since Code 0 is non-CSS. However, the problem is resolved if one choose appropriate sign for each dependent generator. One can check directly that relations $R_1$ and $R_2$ does not generate $-I$.

In order to prove an upper bound on $k$, suppose $7$ does not divide $L$. Then we can show that $ k \leq 4 L $. The proof is very similar to that of Code 1. Let $S$ be the set of generators lying outside a straight tunnel $T$ of length $L$ parallel to $y$-axis, whose cross-section is $1 \times 2$ rectangle enclosed by 6 sites, \ie, the tunnel is $1 \times L \times 2 ~ (\Delta x \times \Delta y \times \Delta z)$. There are $2 L^3 - 4 L$ generators in $S$. We show $S$ is an independent set of generators.

Suppose a linear combination $O$ of generators in $S$ is the identity operator. We show $O$ is the zero combination. Choose the origin of the coordinate system such that the sites enclosing the cross-sectional rectangle are described by $ x = 0,1,\ z=-1,0,1 $. Let $l$ be the set of sites on the straight line given by $ x=z=0$. Since $O$ is the identity, $l$ is acted on by the identity. Every unit edge $e_i$ in $l$ connecting $(0,i,0)$ and $(0,i+1,0)$ is a linear combination of
\begin{align}
 XX &- ZI & &
 ZZ - II  \notag \\
 ZY &- XY & &
 XZ - ZX \label{eq:action_generator_unit_edge_code_0}
\end{align}
which is canceled by the neighboring edges. Since the left-hand side operators of eq.\eqref{eq:action_generator_unit_edge_code_0} are independent, we can unambiguously determine $e_{i+1}$ given $e_i$, \emph{e.g.}, if $e_1 = XX - ZI$, then $e_2 = ZI - XZ$. The right-hand side operators of $e_i$'s form an inference chain $ZI - XZ - ZX - IX - XY - YZ - YY - ZI$. Note that this inference chain exhausts all the combination of operators on the right-hand side of eq.\eqref{eq:action_generator_unit_edge_code_0}. Therefore, $\{ e_i \}_i$ is eventually periodic with period 7, or is eventually $II$. If $7 \nmid L$, we must have $e_i = II$ for all $i$. We showed that the coefficients of the operators touching $l$ in $O$ are all zero. We can repeat the argument to infer that $O$ does not involve stabilizer generators lying in between $z=\pm 1$.

Consider the set of sites on the line $l'$ given by $x=0,z=-1$. A unit edge $e'_i$ in $l'$ connecting $(0,i,-1)$ and $(0,i+1,-1)$ is a linear combination of
\begin{align*}
 ZX &- XZ & &
 XY - ZY  \notag \\
 XX &- ZI & &
 ZZ - II .
\end{align*}
The inference chain is $ZI - IY - XZ - YZ - YX - ZY - XX - ZI$, which is again of eventual period 7. Therefore, $O$ is a zero combination. The direction of the tunnel $T$ can be along any coordinate axis, as is implied by three-fold symmetry.

In a general case where $7$ may divide $L$, we consider three tunnels $T_x, T_y, T_z$ intersecting at one $1 \times 2 \times 2$ box $B$. It is easy to see that stabilizer generators ($S$) lying outside $T_x \cup T_y \cup T_z$ are independent; there are $2 L_x L_y L_z - 8 L_x - 4 L_y - 4 L_z + 16$ independent generators. We can add to $S$ more independent generators. Let the box $B$ be given by $ 0 \leq x \leq 1, ~ -1 \leq y,z \leq 1$. Consider a long contractible tube $T'$ given by $ 1 \leq x \leq L_x, ~ 0 \leq y \leq 1, ~ -1 \leq z \leq 1$. Then $S'$, the union of $S$ and the set of stabilizer generators lying in the tube $T'$, is a set of independent generators. Since $T'$ contains $4(L_x-1)$ generators, $|S'| = 2 L_x L_y L_z - 4 L_x - 4 L_y - 4 L_z + 12$, which proves the upper bound $ k \leq 12L -12$.

\subsection{String logical operators}

There exist string logical operators for Code 0. By three-fold symmetry, these string logical operators form a triple. Using the notation in Table~\ref{tb:string_logical_operators}, they are
\begin{eqnarray*}
\theta^Z_{p} =
 & ZZ[1,0,-1]_{p+\hat{z}} XI[1,0,-1]_{p} ZZ[1,0,-1]_{p-\hat{z}} , \\
\theta^X_{p} = 
 & ZX[-1,1,0]_{p+\hat{x}} XI[-1,1,0]_{p} ZX[-1,1,0]_{p-\hat{x}} , \\
\theta^Y_{p} = 
 & ZY[0,-1,1]_{p+\hat{y}} XI[0,-1,1]_{p} ZY[0,-1,1]_{p-\hat{y}} ,
\end{eqnarray*}
where $\hat{z},\hat{x},\hat{y}$ are directional unit vectors. We call them \emph{basic string logical operators} for Code 0, or basic strings for short.
Applying the techniques of deforming logical string segments, one can show that any logical string segment is a union of `flat' logical string segment and some contiguous part of basic strings. Numerical result suggests that any long flat logical string segment $\zeta^{(0)}_x$ is trivial; it is trivial if the length of $\zeta^{(0)}_x$ is greater than $3w$, where $w$ is the width of $\zeta^{(0)}_x$, for $ w = 2, \ldots, 600 $. Therefore, it is legitimate to conjecture that any long logical string segment is some product of basic strings. If this is true, we can consider a subsystem code by \emph{gauging out} the logical qubits that are affected by basic strings.

Interestingly, we can show that there are at least one logical qubit left in the resulting subsystem code in any finite periodic lattice $\mathbb{Z}_L^3$. To see this, it is enough to calculate the commutation relation of independent basic strings (up to the stabilizer group). After Gram-Schmidt procedure with respect to commutation symplectic form applied to independent set of basic strings, suppose we get $2 k_h$-dimensional hyperbolic space, and $k_i$-dimensional isotropic space. Then the number of gauge qubits is $k_h + k_i$.

It is not hard to see that $\theta^Z_{(0,0,0)} \cdot \theta^X_{(0,0,0)} \cdot \theta^Y_{(0,0,0)}$ is in the stabilizer group $\mathcal{S}_0$ of Code 0 in $\mathbb{Z}_L^3$; it is equal to the product of $Q_0, Q_0^P$ inside the triangle formed by three basic strings. Also, 
\begin{align*}
\theta^Z_{(0,0,0)} \cdot \theta^Z_{(a,b,c)} \in \mathcal{S}_0 , \\
\theta^X_{(0,0,0)} \cdot \theta^X_{(a,b,c)} \in \mathcal{S}_0 , \\
\theta^Y_{(0,0,0)} \cdot \theta^Y_{(a,b,c)} \in \mathcal{S}_0 
\end{align*}
for any $a,b,c \in \mathbb{Z}_L$ such that $ a+b+c = 0$. Therefore the maximally independent set of basic strings up to $\mathcal{S}_0$ is contained in
\[
 \Gamma = \{ \theta^Z_{(0,i,0)}, \theta^X_{(0,0,i)} : i \in \mathbb{Z}_L \},
\]
\ie, $2 k_h + k_i \leq |\Gamma| = 2 L$. If $L$ is odd,
\begin{align*}
 \prod_i \theta^Z_{(0,i,0)} \equiv \sigma_{XI}^{[010]} \in \mathcal{S}_0 ,\\
 \prod_i \theta^X_{(0,0,i)} \equiv \sigma_{XI}^{[001]} \in \mathcal{S}_0 ,
\end{align*}
where $\equiv$ means equality up to $\mathcal{S}_0$. Hence $2 k_h + k_i \leq 2L -2$ for odd $L$.
Note that $\theta^Z_{(0,i,0)}$ and $\theta^Z_{(0,i',0)}$ always commute, and so do $\theta^X_{(0,i,0)}$ and $\theta^X_{(0,i',0)}$. Two basic strings $\theta^X_{(0,i,0)}$ and $\theta^Z_{(0,j,0)}$ anti-commute if and only if they meet at one site. We can write the commutation relation matrix of $\Gamma$ (over the binary field) as $\omega(\Gamma) = \begin{pmatrix} 0 & \omega' \\ \omega'^T & 0 \end{pmatrix}$ where
\[
 \omega'_{ij} = \delta_{[i],[j+2]} + \delta_{[i],[j-2]},
\]
$[i],[j \pm 2]$ are the equivalence classes of integers modulo L, and $\delta$ is the Kronecker delta. Note that $\rank(\omega) = 2 k_h$.
Each row of $\omega'$ is a translation of another, and containes exactly two 1's separated by 4 ($L > 4$).

If $L$ is odd, then $L$ and 4 are relatively prime, and we see that $\rank(\omega') = L-1$. Since $2 k_h + k_i \leq 2L -2$, $k_i=0$ and the number of gauge qubits is $L-1$. If $L$ is even, we need to distinguish two cases. When $4 \mid L$, only $L-4$ rows are independent. It is verified easily if one cyclically rotate the rows of $\omega'$ such that $\omega'_{11}=1$. When $4 \nmid L$, there is a row that contains two 1's that are 2 columns apart, and hence $L-2$ rows become independent. In short, $k_h = \rank \omega' = 4 \lceil L/4 \rceil - 4$, and therefore $ k_h + k_i \leq L + 4$. Since $ k \geq L + 6 q_2(L) $, we conclude that there are at least 1 qubit left after gauging out the basic strings.

\section{Derivation of 2D toric code}

We apply our construction of non-CSS cubic codes to 2D square lattice. We argue that the 2D toric code is the unique stabilizer code under the construction. A square stabilizer generator $s$ has 4 corner operators. $\omega$ is a $4 \times 4$ matrix. Imposing the condition that $s$'s define a stabilizer code, we get the general form:
\begin{equation*}
 \omega = 
 \begin{pmatrix}
  0 & i & j & 0 \\
  i & 0 & 0 & j \\
  j & 0 & 0 & i \\
  0 & j & i & 0
 \end{pmatrix}
\end{equation*}
where $i,j \in \mathbb{F}_2$. In order to ensure a single site operator is the identity, we require that $\rank(\omega) = 2 m $. If $i=j=0$, $\rank(\omega)=0$. If $i=j=1$, then $\rank(\omega)=2$ and the realization of $\omega$ is
\[
 \xymatrix@!0{
  X\ar@{-}[rr]\ar@{-}[dd]&           & Z\ar@{-}[dd]\\
                         & s=s^\text{inv} &              \\
  Z \ar@{-}[rr]          &           & X           
 }
\]

This is the generator for $45^\circ$-rotated 2D toric code.
If $i=1,j=0$, then $\rank(\omega) = 4$ and a realization of $\omega$ is
\[
 \xymatrix@!0{
  XI\ar@{-}[rr]\ar@{-}[dd]&   & ZI\ar@{-}[dd]& & IX\ar@{-}[rr]\ar@{-}[dd]&   & IZ\ar@{-}[dd]\\
                          & s &              & &                         &s^\text{inv}& \\
  IZ \ar@{-}[rr]          &   & IX           & & ZI \ar@{-}[rr]          &   & XI
 }
\]
If $i=0,j=1$, we get the same realization of $\omega$ that is $90^\circ$-rotated. Note that first qubits in the odd numbered rows interact only with second qubits in the even numbered rows, and vice versa. Thus, the code is equal to the non-interacting two copies of 2D toric code. If the periodic lattice has odd linear size, it is a doubly-folded toric code.

\section{Numerical methods}
\label{section:numerical_methods}

We describe our algorithm calculating $k$. As we have seen from proofs of the upper bounds for $k$ (Table~\ref{tb:lower_upper_bound_k}), we need to know the number of independent stabilizer generators in a given finite periodic lattice. Since there are two qubits per site, a Pauli operator on $\mathbb{Z}_L$ is expressed as a $4L^3$-component binary vector. Given that the Pauli operator is $Z$- or $X$-type, the number of components is divided by two. Since there are $2L^3$ stabilizer generators, we need to calculate the rank of $2 L^3 \times 4 L^3$ binary matrix $U$. The Gauss elimination is an efficient algorithm, but the matrix is quite large. A naive sparse matrix method would not be of much help because $U$ may get a large number of non-zero components as Gauss elimination algorithm proceeds.

The proof of the upper bound gives a natural order of the sites and generators such that a large sparse submatrix of $U$ is in a row echelon form; order the sites and generators such that the independence of generators used in the proof of the upper bound is evident from the form of $U$. In this way, we can represent $U$ using memory size $O(L^4)$. It is also advantageous in view of time complexity. Note that since the field is binary, there is no multiplication. The naive Gauss elimination would require $O(L^9)$ additions, while the better representation of $U$ needs only $O(L^5)$ additions.

\end{document}